\documentclass[a4paper]{article}
\pdfoutput=1
\usepackage{amsfonts,amstext,amsmath,amssymb,stmaryrd,mathrsfs,calc,amsthm,mathtools,bold-extra}
\usepackage{comment}
\includecomment{vlong}\excludecomment{vshort}

\newcommand{\eqdef}{\stackrel{{\text{\tiny def}}}{=}}
\newcommand{\F}{\mathbf{F}}
\newcommand{\FGH}[1]{\mathscr{F}_{\!#1}} %
\newcommand{\overto}[1]{\xrightarrow{\!\!{#1}\!\!}}
\newcommand{\step}[1]{\overto{#1}} %

\renewcommand{\vec}[1]{\mathbf{#1}}
\newcommand{\embeds}{\sqsubseteq}
\newcommand{\embedd}{\sqsupseteq}
\newcommand{\comp}{\fatsemi}
\newcommand{\Nat}{\mathbb N}
\newcommand{\pure}{\mathsf{p}}
\newcommand{\purify}{\pure}
\newcommand{\tup}[1]{\langle #1\rangle}
\newcommand{\sep}{\shortmid}
\newcommand{\pad}{\bot}
\newcommand{\tally}{\#}
\newcommand{\trc}{\circledast}
\newcommand{\tc}{\oplus}
\newcommand{\ks}{\ast}
\newcommand{\kp}{+}
\newcommand{\Fw}{\mathsf{Fw}}
\newcommand{\Bw}{\mathsf{Bw}}
\newcommand{\Sim}{\mathsf{Sim}}
\newcommand{\Init}{\mathsf{Init}}
\newcommand{\Fin}{\mathsf{Fin}}
\newcommand{\Seqs}{\mathrm{Seqs}}
\newcommand{\atomi}{\sqsubset_1}
\newcommand{\atoml}{\sqsupset_1}

\theoremstyle{plain}
\newtheorem{lemma}{Lemma}
\newtheorem{theorem}{Theorem}
\newtheorem{proposition}{Proposition}

\newtheorem{fact}{Fact}
\theoremstyle{definition}

\theoremstyle{remark}

\newtheorem{claim}[theorem]{Claim}

\makeatletter%
\newenvironment{decpb}[2]{%
  \smallskip\par\noindent%
  \textbf{#1}~(#2)\hfill%
  \list{}{%
    \labelwidth\z@ \itemindent-\leftmargin%
    \parskip\z@%
    \itemsep\z@%
    \topsep 2.5pt%
    \partopsep\z@%
    \parskip\z@%
    \parsep\z@}%
}{\endlist\par\smallskip}%
\usepackage{subfig,tikz}
\usetikzlibrary{positioning}
\usetikzlibrary{arrows}
\usetikzlibrary{automata}
\usetikzlibrary{trees}
\makeatletter
\def\@listi{\leftmargin\leftmargini
               \topsep 3\p@ \@plus\p@ \@minus\p@
               \parsep 2\p@ \@plus\p@ \@minus\p@
               \itemsep \parsep}
\makeatother
\usepackage{natbib}\renewcommand{\cite}{\citep}
\usepackage{url,hyperref}
\makeatletter\AtBeginDocument{\hypersetup{pdftitle=\@title,pdfauthor=\@author}}\makeatother

 \title{The Parametric Ordinal-Recursive Complexity\\of
  Post Embedding Problems%
  \thanks{Research partially funded by the ANR ReacHard project (ANR
    11 BS02 001 01).  The first author is partially funded by the Tata
    Consultancy Service.  Part of this research was conducted while
    the second author was visiting the Department of Computer Science
    at Oxford University thanks to a grant from the ESF \emph{Games
      for Design and Verification} activity.}}
 \author{Prateek Karandikar$^{\text{1,2}}$\quad Sylvain Schmitz$^{\text{2}}$}
\date{$^{\text{1}}$CMI, Chennai, India\\$^{\text{2}}$LSV, ENS Cachan
  \& CNRS, France}
\begin{document}
\maketitle
\begin{abstract}
 \emph{Post Embedding Problems} are a family of decision problems
 based on the interaction of a rational relation with the subword
 embedding ordering, and are used in the literature to prove non
 multiply-recursive complexity lower bounds.  We refine the
 construction of \citeauthor{lcs} (LICS 2008) and prove parametric
 lower bounds depending on the size of the alphabet.
\end{abstract}
\section{Introduction}
\paragraph{Ordinal Recursive} functions and subrecursive
hierarchies~\citep{rose84,fairtlough98} are employed in computability
theory, proof theory, Ramsey theory, rewriting theory, etc.\ as tools
for bounding derivation sizes and other objects of very high
combinatory complexity.  A standard example is the ordinal-indexed
\emph{extended Grzegorczyk hierarchy} $\FGH{\alpha}$~\citep{lob70},
which characterizes classical classes of functions: for
instance, $\FGH{2}$ is the class of elementary functions,
$\bigcup_{k<\omega}\FGH{k}$ of primitive-recursive ones, and
$\bigcup_{k<\omega}\FGH{\omega^k}$ of multiply-recursive ones.  
Similar tools are required for the classification of decision problems
arising with verification algorithms and logics, prompting the still
young investigation of \emph{fast-growing complexity} classes
$\F_\alpha$ and their associated complete
problems~\citep{friedman99,wqo}.

\paragraph{Post Embedding Problems} (PEPs) have been introduced by
\citet{pepreg} as a tool to prove the decidability of safety and
termination problems in unreliable channel systems.  The most
classical instance of a PEP is called ``regular'' by \citet{pepreg}, but
we will follow \citet{barcelo12} and rather call it \emph{rational} in
this paper:
\begin{decpb}{Rational Embedding Problem}{EP[Rat]}
  \item[input] A rational relation $R$ in
    $\Sigma^\ast\times\Sigma^\ast$.
  \item[question] Is the relation $R\cap{\embeds}$ empty?
\end{decpb}\noindent
Here, the $\embeds$ relation denotes the \emph{subword embedding}
ordering, which relates two words $w$ and $w'$ if $w=c_1\cdots c_n$
and $w'=w_0c_1w_1\cdots w_nc_nw_{n+1}$ for some symbols $c_i$ in
$\Sigma$ and words $w_i$ in $\Sigma^\ast$; in other words, $w$ can be
obtained from $w'$ by ``losing'' some symbol occurrences (maybe none).

Although PEPs appear naturally in relation with channel
systems~\citep{pepreg,CS-concur08,jks-ifiptcs12} and queries on graph
databases~\citep{barcelo12}, their main interest lies in their use in
lower bound proofs for other, sometimes seemingly distantly related
problems \citep{mtl,ata,tsoreach}: in spite of their simple
formulation, they are known to be of non multiply-recursive complexity
in general.  In fact, this motivation has been present from their
inception in \citep{pepreg}: find a ``master'' decision problem
complete for $\F_{\omega^\omega}$, the class of
\emph{hyper-Ackermannian} problems, solvable with non
multiply-recursive complexity, but no less---much like SAT is often
taken as the canonical \textsc{NPTime}-complete problem, or the Post
Correspondence Problem for $\Sigma^0_1$.  This has also prompted a
wealth of research into variants and related questions
\begin{vshort}
\citep{CS-icalp10,barcelo12,KS-csr12}.
\end{vshort}
\begin{vlong}
\citep{CS-fossacs08,CS-blockers,CS-icalp10,barcelo12,KS-csr12}.
\end{vlong}

\medskip In this paper, we revisit and simplify the original proof of
\citet{lcs} that established the hardness of PEPs, and prove tight
\emph{parameterized} lower bounds when the size of the alphabet
$\Sigma$ is fixed.  More precisely, we show that the $(k+2)$-rational
embedding problem, i.e.\ the restriction of EP[Rat] to alphabets
$\Sigma$ of size at most $k+2$, is hard for $\F_{\omega^{k}}$ the
class of \emph{$k$-Ackermannian problems} if $k\geq 2$.  As the
problem can be shown to be in
$\F_{\omega^{k+1}+1}$~\citep{SS-icalp2011,KS-csr12}, %
we argue this to be a rather tight bound.  %
The hyper-Ackermannian lower bound of
$\F_{\omega^\omega}$ first proven by \citeauthor{lcs} then arises when
$|\Sigma|$ is not fixed but depends on the instance.

\begin{figure}[tbp]
  \centering
  \begin{tikzpicture}[->,>=stealth',shorten >=1pt,inner
    sep=1pt,semithick,auto,every node/.style={font=\footnotesize},node
    distance=.6cm]
    \node(F){\textsc{Space} $F_{\omega^k}$};
    \node[right=1.5cm of F](lrr){$(k+2)$-LR[1-bld]};
    \node[above right=.5cm and 1.5cm of lrr](ratep){$(k+2)$-EP[Rat]};
    \node[above right=-.1cm and 1.5cm of lrr](regep){$(k+3)$-EP[Sync]};
    \node[below right=-.1cm and 1.5cm of lrr](lcs){$(k+2)$-LCS};
    \node[below right=.5cm and 1.5cm of lrr](term){$(k+2)$-LT[1-bld]};
    \path[every node/.style={font=\tiny,very near end}]
      (F) edge node[midway]{Prop.~\ref{prop-lrr}} (lrr)
      (lrr.east) edge[bend left=20] node{Prop.~\ref{prop-ratep}} (ratep.west)
      (lrr.east) edge node{Prop.~\ref{prop-regep}} (regep.west)
      (F.east) edge[swap,bend right=10] node{Prop.~\ref{prop-term}} (term.west)
      (lrr.east) edge[swap] node{Prop.~\ref{prop-lcs}} (lcs.west);
  \end{tikzpicture}
  \caption{\label{fig-reds}Relationships between PEPs and similar
    decision problems.}
\end{figure}
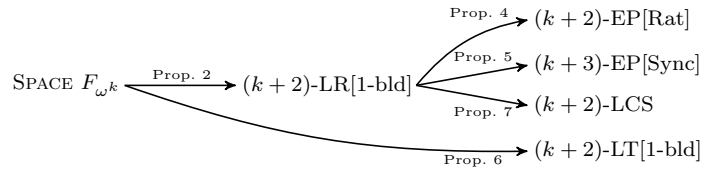
\paragraph{Overview.}
Technically, our results rely on an
implementation of the computations for the \emph{Hardy
  functions} $H^{\omega^{\omega^k}}$ and their inverses
by successive applications of a relation with a fixed \emph{bounded
  length discrepancy}.  The main difficulty here is that this
implementation should be \emph{robust} for the symbol losses
associated with the embedding relation.  It requires in particular a
robust encoding of ordinals below $\omega^{\omega^k}$ as sequences
over an alphabet of $k+2$ symbols, for which we adapt the
constructions of \citep{lcs,haddad12}; see \autoref{sec-hardy}.

This allows us to show in \autoref{sec-ack} that the following problem
is $\F_{\omega^k}$-hard when $k\geq 2$, $|\Sigma|=k+2$, and $R$ has a bounded
length discrepancy of 1:
\begin{decpb}{Lossy Rewriting}{LR[Rat]}
\item[input] A rational relation $R$ in 
  $\Sigma^\ast\times\Sigma^\ast$ and two words $w$ and $w'$
  in $\Sigma^\ast$.
  \item[question] Does $(w,w')$ belong to the reflexive transitive
    closure $R_{\embedd}^\trc$?
\end{decpb}\noindent
Here $R_{\embedd}$ denotes the ``lossy version'' of the relation $R$,
defined formally as the composition ${\embedd}\comp R\comp{\embedd}$.
We denote the restricted problem when the rational relation $R$ has a
bounded length discrepancy of 1 by LR[1-bld].

We then show in \autoref{sec-rel} that LR[1-bld] can easily be reduced
to EP[Rat] and other (parameterized) embedding problems---including
EP[Sync], a restriction of EP[Rat] introduced by \citet{barcelo12}
where the relation $R$ is \emph{synchronous} (aka \emph{regular}), and
which required a complex lower bound proof.  Figure~\ref{fig-reds}
summarizes the lower bounds presented in this paper.  In a sense, LR
is our own champion for the title of ``master'' problem for
$\F_{\omega^\omega}$.  Besides its rather simple statement, note that
the related question of whether $(w,w')$ belongs to $R^\trc$ is
undecidable by an easy reduction from the acceptance problem for Turing
machines.

Let us now turn to the necessary formal background on PEPs in
\autoref{sec-pep}.  Due to space constraints, some proof details will
be found in the appendices.

\section{Post Embedding Problems}\label{sec-pep}
\paragraph{Rational Relations}\citep{rtrans} play an important role
in the following, as they provide a notion of finitely presentable
relations over strings more powerful than string rewrite systems, and
come with a large body of theory and results~\citep[see
e.g.][Chap.~IV]{saka}.  Let us quickly skim over the notations and
definitions that will be needed in this paper.

We assume the reader to be familiar with the basic characterizations
of \emph{rational relations} $R$ between two finite alphabets
$\Sigma$ and $\Delta$ by
\begin{description}
\item[closure] of the finite relations in $\Sigma^\ks\times\Delta^\ks$
  under union, concatenation, and Kleene star,\footnote{We use
  different symbols ``$^\ks$'' and ``$^\kp$'' for Kleene star and
  Kleene plus,
  i.e.\ iteration of concatenation ``$\cdot$'' on the one hand, and
  ``$^\trc$'' and ``$^\tc$'' for reflexive transitive closure and
  transitive closure, i.e.\ iteration of composition ``$\comp$'' on
  the other hand.  Rational relations and length-preserving relations
  are closed under Kleene star, but none of the classes of relations
  we consider is closed under reflexive transitive closure.}
\item[finite transductions] defined by normalized transducers
  $\?T=\tup{Q,\Sigma,\Delta,\delta,I,F}$ where $Q$ is a finite set of
  states, $\delta\subseteq
  Q\times((\Sigma\times\{\varepsilon\})\cup(\{\varepsilon\}\times\Delta))\times
  Q$ ($\varepsilon$ denoting the empty word of length
  $|\varepsilon|=0$), initial set of states $I\subseteq Q$, and final
  set of states $F\subseteq Q$,
\item[decomposition] into a regular language $L$ over some finite
  alphabet $\Gamma$ and two morphisms $u{:}\,\Gamma^\ks\to\Sigma^\ks$
  and $v{:}\,\Gamma^\ks\to\Delta^\ks$ s.t.\
  $R=u^{-1}\comp\mathrm{Id}_L\comp v$, where $\mathrm{Id}_L$ is the
  identity function over the restricted domain $L$.
\end{description}

This last characterization is known as 
\begin{vlong}
\citeauthor{nivat}%
\end{vlong}
\begin{vshort}
Nivat%
\end{vshort}
's Theorem,
and shows that EP[Rat] can be stated alternatively as taking as input
a rational language $L$ in $\Gamma^\ast$ and two morphisms $u$ and $v$
from $\Gamma^\ast$ to $\Sigma^\ast$ and asking whether there exists
some word $x$ in $L$ s.t.\ $u(x)\embeds v(x)$~\citep{pepreg}.  This
justifies the name of ``Post Embedding Problem'', as the
related \emph{Post Correspondence Problem} asks instead given $u$ and
$v$ whether there exists $x$ in $\Gamma^+$ s.t.\ $u(x)=v(x)$.

\paragraph{Synchronous Relations} are a restricted class of rational
relations that display closure under intersection and complement, in
addition to e.g.\ the closure under composition and inverse that all
rational relations enjoy.  A rational relation has \emph{$b$-bounded
length discrepancy} if the absolute value of $|u|-|v|$ is at most $b$
for all $(u,v)$ in $R$, and has \emph{bounded length discrepancy}
(bld) if there exists such a finite $b$. In particular, it
is \emph{length-preserving} if $|u|=|v|$, i.e.\ if it has bld $0$.
A \emph{synchronous relation} is a finite union of relations of form
$\{(u,vw)\mid (u,v)\in R\wedge w\in L\}$ and $\{(uw,v)\mid (u,v)\in
R\wedge w\in L\}$ where $R$ ranges over length-preserving rational
relations and $L$ over regular languages.  In terms of classes of
relations in $\Sigma^\ks\times\Delta^\ks$, we have the strict
inclusions
\begin{equation}
 \mathrm{lp}=0\mathrm{\text-bld}\subsetneq\cdots\subsetneq
 b\mathrm{\text-bld}\subsetneq (b+1)\mathrm{\text-bld}\subsetneq\cdots\subsetneq\mathrm{bld}\subsetneq\mathrm{Sync}\subsetneq\mathrm{Rat}\;.
\end{equation}

\paragraph{Post Embedding Problems,} as we have seen in the
introduction, are concerned with the interplay of a rational relation
$R$ in $\Sigma^\ks\times\Sigma^\ks$ with the subword embedding
ordering $\embeds$.  The latter is a particular case of a
(deterministic) rational relation that is not synchronous.  Both EP[Rat]
and LR[Rat] are particular instances of more general, undecidable
problems: the emptiness of intersection of two rational relations for
EP[Rat], and the word problem in the reflexive transitive closure of a
rational relation for LR[Rat].  We can add another natural problem to
the set of PEPs:
\begin{decpb}{Lossy Termination}{LT[Rat]}
\item[input] A rational relation $R$ over $\Sigma$ and a word $w$ in
  $\Sigma^\ast$.
\item[question] Does $R_\embedd^\trc$ terminate from $w$, i.e.\ is every
sequence
$w=w_0\mathbin{R_\embedd}w_1\mathbin{R_\embedd}\cdots\mathbin{R_\embedd}w_i\mathbin{R_\embedd}\cdots$
with $w_0,w_1,\dots,w_i,\dots$ in $\Sigma^\ast$ finite?
\end{decpb}
Again, this is a variant of the termination problem, which is in
general undecidable when the relation is not lossy.

\paragraph{Restrictions.}  We parameterize PEPs with the subclass of
rational relations under consideration for $R$ and the cardinal of the
alphabet $\Sigma$; for instance, $(k+2)$-EP[Sync] is the variant of
EP[Rat] where the relation is synchronous and $|\Sigma|=k+2$.  We are
interested in this paper in providing $\F_{\omega^k}$ lower
bounds with the smallest possible class of relations and smallest
possible alphabet size, but we should also mention that some (rather
strong) restrictions become tractable:
\begin{itemize}
\item \citet{barcelo12} show that EP[Rec]---where a \emph{recognizable
    relation} is a finite union of products $L\times L'$ where $L$ and
  $L'$ range over regular languages---is in \textsc{NLogSpace},
  because the intersection $R\cap{\embeds}$ is rational, and can
  effectively be constructed and tested for emptiness on the fly,
\item \citet{pepreg} show that EP[2Morph]---where a \emph{2-morphic
    relation}~\citep{latteux83} is the composition $R=(u^{-1}\comp
  v)\setminus\{(\varepsilon,\varepsilon)\}$ of two morphisms $u$ and
  $v$ from $\Gamma^\ks$ to $\Sigma^\ks$---is in \textsc{LogSpace},
  because it reduces to checking whether there exists $a$ in $\Gamma$
  s.t.\ $u(a)\embeds v(a)$,
\item the case EP[Rewr] of \emph{rewrite relations} is
  similarly in \textsc{LogSpace}: a rewrite relation $R$ is defined from a
  \emph{semi-Thue system}, i.e.\ a finite set $\Upsilon$ of rules $(u,v)$ in
  $\Sigma^\ks\times\Sigma^\ks$, as ${\to_\Upsilon}=\{(wuw',wvw')\mid
  w,w'\in\Sigma^\ks\wedge (u, v)\in\Upsilon\}$, and EP[Rewr] reduces to checking
  whether $u\embeds v$ for some rule $(u,v)$ of $\Upsilon$,
\item the unary alphabet case of $1$-EP[Rat] is in \textsc{NLogSpace}:
  this can be seen using Parikh images and Presburger arithmetic; see
  App.~\ref{ax-1ep} for details:
\end{itemize}
\begin{proposition}\label{prop-1ep}
  The problem $1$-EP[Rat] is in \textsc{NLogSpace}.
\end{proposition}

\section{Hardy Computations}\label{sec-hardy}
We use the \emph{Hardy hierarchy} as our main subrecursive
hierarchy~\citep{lob70,rose84,fairtlough98}.  Although we will only
use the lower levels of this hierarchy, its general definition is
worth knowing, as it is archetypal of ordinal-indexed
\emph{subrecursive hierarchies}; see~\cite{wqo} for a self-contained
presentation.

\subsection{The Hardy Hierarchy}\label{sub-hardy}
\paragraph{Ordinal Terms.}
Let $\varepsilon_0$ be the smallest solution of the equation
$\omega^x=x$.  It is well-known that any ordinal
$\alpha<\varepsilon_0$ can be written uniquely in Cantor Normal Form
(CNF) as a sum
\begin{equation}\label{eq-def-cnf}
  \alpha = \omega^{\beta_1}+\cdots+\omega^{\beta_n}
\end{equation}
where $\beta_n\leq\cdots\leq\beta_1<\alpha$ and each $\beta_i$ is
itself in CNF.  This ordinal $\alpha$ is $0$ if $n=0$
in \eqref{eq-def-cnf}, a \emph{successor ordinal} if $\beta_n$ is $0$,
and a \emph{limit ordinal} otherwise.  Subrecursive hierarchies are
defined through assignments of \emph{fundamental sequences}
$(\lambda_n)_{n<\omega}$ for limit ordinals $\lambda<\varepsilon_0$,
satisfying $\lambda_n<\lambda$ for all $n$ and
$\lambda=\sup_n\lambda_n$.  A standard assignment is defined by:
\begin{align}
\label{eq-def-fundseq}
\bigl(\gamma + \omega^{\alpha+1}\bigr)_n & \eqdef \gamma +
\omega^\alpha\cdot n, & \bigl(\gamma + \omega^{\lambda}\bigr)_n
& \eqdef \gamma +\omega^{\lambda_n},
\end{align}
thus verifying $\omega_n = n$.  Let
$\Omega\eqdef\omega^{\omega^\omega}$; this yields for instance
$\Omega_k=\omega^{\omega^k}$ and, if $k>0$,
$(\Omega_k)_{n}=\omega^{\omega^{k-1}\cdot n}$.

\paragraph{Hardy Hierarchy.}
The \emph{Hardy hierarchy} $(H^\alpha)_{\alpha<\varepsilon_0}$ is an
ordinal-indexed hierarchy of functions $H^\alpha{:}\,\+N\to\+N$
defined by
\begin{align}\label{eq-def-hardy}
  H^0(n)&\eqdef n &
  H^{\alpha+1}(n)&\eqdef H^\alpha(n+1) &
  H^{\lambda}(n)&\eqdef H^{\lambda_n}(n)\;.
\end{align}
Observe that $H^1$ is simply the successor function, and more
generally $H^\alpha$ is the $\alpha$th iterate of the successor
function, using diagonalisation to treat limit ordinals.  A related
hierarchy is the \emph{fast growing hierarchy}
$(F_{\alpha})_{\alpha<\varepsilon_0}$, which can be defined by
$F_\alpha\eqdef H^{\omega^\alpha}$, resulting in $F_0(n)=H^1(n)=n+1$,
$F_1(n)=H^\omega(n)=H^{n}(n)=2n$, $F_2(n)=H^{\omega^2}(n)=2^nn$ being
exponential, $F_3=H^{\omega^3}$ being non-elementary,
$F_\omega=H^{\omega^\omega}=H^{\Omega_1}$ being an Ackermannian
function, $F_{\omega^k}=H^{\Omega^k}$ a $k$-Ackermannian function, and
$F_{\omega^\omega}=H^{\Omega}$ an hyper-Ackermannian function.

\paragraph{Fast-Growing Complexity Classes.}
Our intention is to establish the ``$F_{\omega^k}$ hardness'' of Post
embedding problems.  In order to make this statement more precise, we
define the class $\F_{\omega^k}$ of \emph{$k$-Ackermannian problems}
as a specific instance of the \emph{fast-growing complexity classes}
defined for $\alpha\geq 3$ by
\begin{vshort}
\pagebreak
\end{vshort}
\begin{align}
  \F_\alpha&\eqdef\bigcup_{p\in\bigcup_{\beta<\alpha}\FGH{\beta}}\text{\textsc{DTime}}(F_\alpha(p(n)))\;,&
  \FGH{\alpha}&=\bigcup_{c<\omega}\text{\textsc{FDTime}}(F^c_\alpha(n))\;,
\end{align}
where $\FGH{\alpha}$ defined above is the $\alpha$th level of the \emph{extended Grzegorczyk hierarchy}
\citep{lob70} when $\alpha\geq 2$.  The classes $\F_\alpha$ are naturally equipped with
$\bigcup_{\beta<\alpha}\FGH{\beta}$ as class of reductions.  For
instance, because $\bigcup_{k<\omega}\FGH{\omega^k}$ is exactly the
set of multiply-recursive functions, $\F_{\omega^\omega}$ captures the
intuitive notion of hyper-Ackermannian problems closed under
multiply-recursive reductions.\footnote{Note that, at such high
  complexities, the usual distinctions between deterministic vs.\
  nondeterministic, or time-bounded vs.\ space-bounded computations
  become irrelevant.
  In particular, $\FGH{2}$ is the set of elementary
  functions, and $\F_3$ the class of problems with a tower of
  exponentials of height bounded by some elementary function of the
  input as an upper bound.}

\paragraph{Hardy Computations.}
The fast-growing and Hardy hierarchies have been used in several
publications to establish Ackermannian and higher complexity
bounds \citep{lcs,SS-icalp2011,haddad12,wqo}.  The principle in
their use for lower bounds is to view \eqref{eq-def-hardy}, read
left-to-right, as a rewrite system over $\varepsilon_0\times\+N$, and
later implement it in the targeted formalism.
Formally, a (forward) \emph{Hardy computation} is a sequence
\begin{equation}\label{eq-hcomp}
  \alpha_0,n_0\step{}\alpha_1,n_1\step{}
  \alpha_2,n_2\step{}\cdots\step{}\alpha_\ell,n_\ell
\end{equation}
of evaluation steps implementing the equations in \eqref{eq-def-hardy}
seen as left-to-right rewrite rules over \emph{Hardy configurations}
$\alpha,n$.  It guarantees $\alpha_0>\alpha_1>\alpha_2>\cdots$ and
keeps $H^{\alpha_i}(n_i)$ invariant.  We say it is \emph{complete}
when $\alpha_\ell=0$ and then $n_\ell%
=H^{\alpha_0}(n_0)$ (we also consider incomplete computations).
A \emph{backward Hardy computation} is obtained by
using \eqref{eq-def-hardy} as right-to-left rules.  For instance,
\begin{equation}\label{eq-hardy-Omega}
  \omega^{\omega^k},n\to\omega^{\omega^{k-1}\cdot
    n},n\to\omega^{\omega^{k-1}\cdot (n-1)+\omega^{k-2}\cdot n},n
\end{equation}
constitute the first three steps of the forward Hardy computation
starting from $\Omega_k,n$ if $k>1$ and $n>0$.

\paragraph{Termination of Hardy Computations.}  Because
$\alpha_0>\alpha_1>\cdots>\alpha_\ell$ in a forward Hardy computation
like \eqref{eq-hcomp}, it necessarily terminates.  For inverse
computations, this is less immediate, and we introduce for this a
\emph{norm} $\|\alpha\|$ of an ordinal $\alpha$ in $\varepsilon_0$ as
its count of ``$\omega$'' symbols when written as an ordinal term:
formally, $\|.\|{:}\,\varepsilon_0\to\+N$ is defined by
\begin{align}
  \|0\|&\eqdef 0 &\|\omega^{\alpha}\|&\eqdef
  1+\|\alpha\|&\|\alpha+\alpha'\|&\eqdef\|\alpha\|+\|\alpha'\|\;.
\end{align}
\begin{vlong}
Observe that $\|\alpha\cdot m\|=\|\alpha\|\cdot m$.
\end{vlong}
We can check 
\begin{vlong}
by transfinite induction on $\alpha>0$ 
\end{vlong}
that, for any limit ordinal $\lambda$, %
$\|\lambda_n\|>\|\lambda\|$ whenever $n>1$.
\begin{vlong}
Indeed, if $\alpha=\beta+1$,
then $\|\lambda_n\|=\|\gamma+\omega^\beta\cdot
n\|=\|\gamma\|+(1+\|\beta\|)n>\|\gamma\|+2+\|\beta\|=\|\lambda\|$, and
in the limit case,
$\|\lambda_n\|=\|gamma\|+1+\|\alpha_n\|>\|\gamma\|+1+\|\alpha\|=\|\lambda\|$
by ind.\ hyp.
\end{vlong}
Therefore, 
\begin{vlong}
if $n$ is larger than $1$ in a configuration $\alpha,n$ of
an inverse Hardy computation following \eqref{eq-def-hardy} from right
to left, either we apply the successor rule and reach $\alpha+1,n-1$
with a decreased $n$, or we apply the limit rule and reach $\alpha',n$
s.t.\ $\alpha=\alpha'_n$ with a decreased $\|\alpha\|$:
\end{vlong}
in a backward
Hardy computation, the pair $(n,\|\alpha\|)$ decreases for the
lexicographic ordering over $\+N^2$.  As this is a well-founded
ordering, we see that backward computations terminate if $n$ remains
larger than $1$---which is a reasonable hypothesis for the following.

\subsection{Encoding Hardy Configurations}\label{sub-hardy-conf}
Our purpose is now to encode Hardy computations as relations over
$\Sigma^\ks$.  This entails in particular (1)~encoding configurations
$\alpha,n$ in $\Omega_k\times\+N$ of a Hardy computation as finite
sequences using \emph{cumulative ordinal descriptions} or
``\emph{codes}'', which we do in this subsection, and (2)~later
in \autoref{sub-hardy-comp} designing a 1-bld relation that implements
Hardy computation steps over codes.  A constraint on codes is that
they should be \emph{robust} against losses, i.e.\ if $\pi(x)$ and
$\pi(x')$ are the ordinals associated to the codes $x$ and $x'$, then
$H^{\pi(x)}(n)\leq H^{\pi(x')}(n)$---pending some hygienic conditions
on $x$ and $x'$, see Lem.~\ref{lem-robust}.

\paragraph{Finite Ordinals} below $k$ can be represented as single
symbols $a_0,\dots,a_{k-1}$ of an alphabet $\Sigma_k$ along with a
bijection
\begin{equation}\label{eq-def-finite}
  \varphi(a_i)\eqdef i\;.
\end{equation}

\paragraph{Small Ordinals} below $\omega^{k}$ are then easily encoded as
finite words over $\Sigma_k$: given a word $w=b_1\cdots b_n$ over
$\Sigma_k$, we define its associated ordinal in $\omega^{k}$ as
\begin{equation}\label{eq-def-beta}
  \beta(w)\eqdef \omega^{\varphi(b_1)}+\cdots+\omega^{\varphi(b_n)}\;.
\end{equation}
Note that $\beta$ is surjective but not injective: for instance,
$\beta(a_0a_1)=\beta(a_1)=\omega$.  By restricting ourselves to
\emph{pure} words over $\Sigma_k$, i.e.\ words satisfying
$\varphi(b_j)\geq\varphi(b_{j+1})$ for all $1\leq j<n$, we obtain a
bijection between $\omega^k$ and $\pure(\Sigma_k^\ast)$ the set of
pure finite words in $\Sigma_k^\ks$, because then \eqref{eq-def-beta}
is the CNF of $\beta(w)$.

\paragraph{Large Ordinals} below $\Omega_k$ are denoted
by \emph{codes}~\citep{lcs,haddad12}, which are $\tally$-separated
words over the extended alphabet
$\Sigma_{k\tally}\eqdef\Sigma_k\uplus\{\tally\}$.  A code $x$ can be
seen as a concatenation $w_1\tally w_2\tally\cdots\tally w_p\tally
w_{p+1}$ where each $w_i$ is a word over $\Sigma_k$.  Its
associated ordinal $\pi(x)$ in $\Omega_k$ is then defined as
\begin{gather}
  \pi(x)\eqdef\omega^{\beta(w_1w_2\cdots w_p)}
              +\cdots+\omega^{\beta(w_1w_2)}+\omega^{\beta({w_1})}\;,
  \shortintertext{or inductively by}
  \pi(w)\eqdef 0,\qquad\pi(w\tally
              x)\eqdef\omega^{\beta(w)}\cdot\pi(x)+\omega^{\beta(w)}
\end{gather}
for $w$ a word in $\Sigma^\ks_k$ and $x$ a code.
For instance, $\pi(a_1a_0\tally)=\omega^{\omega+1}=\pi(a_0a_1a_0\tally
a_3)$, or, closer to our concerns, the initial ordinal in our
computations is $\pi(a_{k-1}^n\tally)=(\Omega_k)_n$ when $k>0$.

Observe that $\pi$ is surjective, but not injective.  We can mend this
by defining a \emph{pure} code $x=w_1\tally\cdots\tally w_p\tally
w_{p+1}$ as one where $w_{p+1}=\varepsilon$ and every word $w_i$ for
$1\leq i\leq p$ is pure---note that it does not force the
concatenation of two successive words $w_iw_{i+1}$ of $x$ to be pure.
This is intended, as this is the very mechanism that allows $\pi$ to
be a bijection between $\Omega_k$ and $\pure(\Sigma_{k\tally
}^*)$ (see App.~\ref{app-pure}):
\begin{lemma}\label{lem-pure}%
  The function $\pi$ is a bijection from $\pure(\Sigma_{k\tally}^\ks)$
  to $\Omega_k$.
\end{lemma}
We also define $\pure(x)$ to be the unique pure code $x'$
verifying $\pi(x)=\pi(x')$; then $\pure(x)\embeds x$, and $x\embeds x'$
implies $\pure(x)\embeds\pure(x')$.

\paragraph{Hardy Configurations} $\alpha,n$ are finally encoded as
sequences $c=\pi^{-1}(\alpha)\sep\tally^n$ using a separator
``$\sep$'', i.e.\ as sequences in the language
$\mathrm{Confs}\eqdef\pure(\Sigma_{k\tally}^\ks)\cdot\mbox{$\{\sep\}$}\cdot\{\tally\}^\ks$.
This is a regular language over $\Sigma_{k\tally}\uplus\{\sep\}$, but the most
important fact about this encoding is that it is \emph{robust} against
symbol losses as far as the corresponding computed Hardy values are
concerned.  Robustness is a critical part of hardness proofs based on
Hardy functions.  The main difficulty rises from the
fact that the Hardy functions are not monotone in their ordinal
parameter: for instance, $H^{\omega}(n)=H^{n}(n)=2n$ is less than
$H^{n+1}(n)=2n+1$.  Code robustness is addressed in
\citep[Prop.~4.3]{lcs}%
\begin{vlong}
, and in \citep[Prop.~16]{haddad12} for a more complex encoding of
ordinals below $\omega^{\omega^{\omega^k}}$ as vector sequences%
\end{vlong}
.  Robustness is the limiting factor that prevents us from reducing
languages in $\F_{\alpha}$ for $\alpha>\Omega$ into PEPs.
\begin{lemma}[Robustness]\label{lem-robust}
Let $c=x\sep\tally^n$ and $c'=x'\sep\tally^{n'}$ be two sequences in
$\mathrm{Confs}$.  If $c \embeds c'$, then $H^{\pi(x)}(n)\leq
H^{\pi(x')}(n')$.
\end{lemma}

\subsection{Encoding Hardy Computations}\label{sub-hardy-comp}
It remains to present a 1-bld relation that implements Hardy
computations over Hardy configurations encoded as sequences in
$\mathrm{Confs}$. 
We translate the equations from \eqref{eq-def-hardy} into a relation
$R_H=(R_0\cup R_1\cup R_2)\cap(\mathrm{Confs}\times\mathrm{Confs})$, which can be reversed for
backward computations:
\begin{align}
\!\!\!\!R_0&\!\eqdef\!\{(\tally x\sep \tally^{n}, x\sep \tally^{n+1})\mid n\geq
0,x\in\Sigma_{k\tally}^\ks\}\label{eq-rz0}\\
\!\!\!\!R_1&\!\eqdef\!\{(wa_0\tally x\sep \tally^{n},w\tally^n \purify(a_0
x)\sep\tally^n)\mid n>1, w\in\Sigma_k^\ks, x\in\Sigma_{k\tally}^\ks\}\label{eq-rz1}\\
\!\!\!\!R_2&\!\eqdef\!\{(wa_i\tally x\sep \tally^{n},wa_{i-1}^{n}\tally\purify(a_i x)\sep\tally^n)\mid n>1, i > 0,
w\in\Sigma_k^\ks, x\in\Sigma_{k\tally}^\ks\}\!\!\label{eq-rz2}
\end{align}
The relation~$R_0$ implements the successor case, while
$R_1$ and $R_2$ implement the limit case of
\eqref{eq-def-fundseq} for ordinals of form $\gamma+\omega^{\alpha+1}$
and $\gamma+\omega^{\lambda}$ respectively.  The restriction to $n>1$
in~$R_1$ and~$R_2$ enforces termination for backward
computations; it is not required for correctness.  Because $R_H$ is a
direct translation of \eqref{eq-def-hardy} over $\mathrm{Confs}$:
\begin{lemma}[Correctness]\label{lem-crct}
  For all $\alpha,\alpha'$ in $\Omega_k$ and $n,n'>1$,
  $\mbox{$(\pi^{-1}(\alpha)\sep\tally^n)$}(R_H\cup R_H^{-1})^\trc(\pi^{-1}(\alpha')\sep\tally^{n'})$
  iff $H^\alpha(n)=H^{\alpha'}(n')$.
\end{lemma}%

Unfortunately, although $R_0$ is a length-preserving rational
relation, $R_1$ and $R_2$ are not 1-bld, nor even rational.  However,
they can easily be broken into smaller steps, which are rational---as
we are applying a reflexive transitive closure, this is at no expense
in generality.  This requires more complex encodings of Hardy
configurations, with some ``finite state
control'' and a working space in order to keep track of where we are in
our small steps.  Because we do not want to spend new symbols in this
encoding, given some finite set $Q$ of states, we work on sequences in
\begin{equation}
  \Seqs\eqdef\{a_0,a_1\}^{\lceil\log
  |Q|\rceil}\cdot\{\sep\}\cdot\pure(\Sigma_k^\ks)\cdot\{\tally\}^\ks\cdot\{\sep\}\cdot\pure(\Sigma^\ks_{k\tally})\cdot\{\sep\}\cdot\{\tally,a_0,a_1\}^\ks\;.
\end{equation}
with four segments separated by ``$\sep$'': a state, a working segment, an
ordinal encoding, and a counter.  Given a state $q$ in
$Q$, we use implicitly its binary encoding as a sequence of fixed
length over $\{a_0,a_1\}$.
\begin{vlong}
  Our sequences in ``normal'' mode look like
``$q\sep\sep\pi^{-1}(\alpha)\sep\tally^n$'' with an empty working
segment and only $\tally$'s as counter symbols.
\end{vlong}

We define two relations $\Fw$ and $\Bw$ with domain and range $\Seqs$
that implement forward and backward computations with $R_H$.  A typical
case is that of computations with $R_1$, which can be implemented as
the closure of the union:
\begin{align}
  q_\Fw\sep \sep wa_0\tally x\sep \tally^{n+2}
  &\;\mathbin{\Fw_1}\;q_{\Fw_1}\sep
  w\tally\sep\purify(a_0x)\sep\tally^{n+1}a_0\label{eq-rFw10}\\
  q_{\Fw_1}\sep w\tally^m\sep x\sep\tally^{n+1}a_0^{p+1}
  &\;\mathbin{\Fw_1}\;q_{\Fw_1}\sep w\tally^{m+1}\sep x\sep\tally^{n}a_0^{p+2}
  \label{eq-rFw11}\\
  q_{\Fw_1}\sep w\tally^{m+1}\sep x\sep a_0^{n+2}
  &\;\mathbin{\Fw_1}\;q_{\Fw_1}\sep\sep w\tally^{m+1} x\sep\tally^{n+2}
  \label{eq-rFw13}
\end{align}
for $m,n,p$ in $\+N$, $w$ in $\pure(\Sigma_{k}^\ks)$, and $x$ in
$\pure(\Sigma_{k\#}^\ks)$.  Note that $\purify(a_0x)$ returns
$a_0x$ if $x$ begins with $\tally$ or $a_0$, and $x$ otherwise.  The
corresponding backward computation for $R_1$ inverses the relations
in~(\ref{eq-rFw10}--\ref{eq-rFw13}) and substitutes $q_\Bw$ and
$q_{\Bw_1}$ for $q_\Fw$ and $q_{\Fw_1}$.  The reader should be able to
convince herself that this is indeed feasible in a rational 1-bld
fashion; for instance, \eqref{eq-rFw11} can be written as a rational
expression
\begin{gather}
  \begin{bmatrix}q_{\Fw_1}{\sep}\\ q_{\Fw_1}{\sep}\end{bmatrix}\cdot
  \mathrm{Id}_{\Sigma_k^\ks}\cdot
  \begin{bmatrix}\tally\\\tally\end{bmatrix}^{\!\ks}\cdot
  \begin{bmatrix}\varepsilon\\\tally\end{bmatrix}\cdot
  \begin{bmatrix}\sep\\\sep\end{bmatrix}\cdot
  \mathrm{Id}_{\Sigma_{k\tally}^\ks}\cdot
  \begin{bmatrix}\sep\\\sep\end{bmatrix}\cdot
  \begin{bmatrix}\tally\\\tally\end{bmatrix}^{\!\ks}\cdot
  \begin{bmatrix}\tally\\\varepsilon\end{bmatrix}\cdot
  \begin{bmatrix}a_0\\a_0\end{bmatrix}^{\!\kp}\cdot
  \begin{bmatrix}\varepsilon\\a_0\end{bmatrix}\,.
\end{gather}
\begin{vlong}
A full description of $\Fw$ and $\Bw$ can be found in
App.~\ref{ax-rules}.  
\end{vlong}

Observe that separators ``$\sep$'' are reliable, and that losses
cannot pass unnoticed in the constant-sized state segment of a
sequence in $\Seqs$; thus we can use lemmas~\ref{lem-robust}
and~\ref{lem-crct} to prove that $\Fw_\embedd^\trc$ and
$\Bw_\embedd^\trc$ are ``weak'' implementations of $H^\alpha$ and its
inverse when $\alpha$ is in $\Omega_k$.  Not any reformulation of
$R_H$ as the closure of a rational relation would work here: our
relation also needs to be robust to losses; see App.~\ref{ax-rules}
for details.
\begin{lemma}[Weak Implementation]\label{lem-corr}
  The relations $\Fw$ and $\Bw$ are 1-bld and terminating.
  Furthermore, if $k\geq 1$, $m,n>1$ and $\alpha\in\Omega_k$,
  \begin{align*}
  (q_\Fw\sep\sep\pi^{-1}(\alpha)\sep\tally^n)\:&\Fw_\embedd^\trc\:(q_\Fw\sep\sep\sep\tally^m)&&\!\!\text{implies $m\leq H^\alpha(n)$}\\
  (q_\Bw\sep\sep\sep\tally^m)\:&\Bw_\embedd^\trc\:(q_\Bw\sep\sep\pi^{-1}(\alpha)\sep\tally^n)&&\!\!\text{implies $m\geq H^{\alpha}(n)$}
\end{align*}
and there exists rewrites verifying $m=H^\alpha(n)$ in both of the above cases.
\end{lemma}

\section{The Parametric Complexity of LR[1-bld]}\label{sec-ack}
Now equipped with suitable encodings for Hardy computations, we can
turn to the main result of the paper: Prop.~\ref{prop-lrr} below shows
the $\F_{\omega^k}$-hardness of \mbox{$(k+2)$-LR[1-bld]}.  As we
obtain almost matching upper bounds in \autoref{sub-up}, we deem this
to be rather tight.

\subsection{Lower Bound}\label{sub-low}

Thanks to the relations over $\Sigma_{k\tally}\uplus\{\sep\}$ defined
in \autoref{sec-hardy}, we know that we can weakly compute with $\Fw$ a
``budget space'' as a unary counter of size $F_{\omega^k}(n)$, and
later check that this budget has been maintained by running through
$\Bw$.  We are going to insert the simulation of an
$\F_{\omega^k}$-hard problem between these two phases of budget
construction and budget verification, thereby constructing
$\F_{\omega^k}$-hard instances of $(k+2)$-LR[1-bld].
\begin{proposition}\label{prop-lrr}
  Let $k\geq 2$.  Then $(k+2)$-LR[1-bld] is $\F_{\omega^k}$-hard.
\end{proposition}

\paragraph{Bounded Semi-Thue Reachability.}  The problem we reduce
from is a space-bounded variant of the \emph{semi-Thue reachability problem}
(aka \emph{semi-Thue word problem}):
as already mentioned in
\autoref{sec-pep}, a \emph{semi-Thue system} $\Upsilon$ over an
alphabet is a finite set of rules $(u,v)$ in
$\Sigma^\ks\times\Sigma^\ks$, defining a \emph{rewrite relation}
$\to_\Upsilon$.
\begin{vshort}
  The semi-Thue reachability problem, or R[Rewr], is a reliable
  version of the lossy reachability problem.
\end{vshort}
\begin{vlong}
\begin{decpb}{Semi-Thue Reachability}{R[Rewr]}
\item[input] A semi-Thue system $\Upsilon$ over an alphabet $\Sigma$,
  and words $y$ and $y'$ in $\Sigma^\ks$.
\item[question] Is it the case that $y\to^\trc_\Upsilon y'$?
\end{decpb}
\end{vlong}
This problem is in general undecidable, as one can express the ``next
configuration'' relation of a Turing machine as a semi-Thue system.
Its \emph{$F_{\omega^k}$-bounded} version for some $k\geq 1$ takes as
input an instance $\tup{\Upsilon,y,y'}$ of size $n$  where, if 
$y\to^\trc_\Upsilon x$, then $|x|\leq F_{\omega^k}(n)$.
This is easily seen to be hard for $\F_{\omega^k}$, even for a
binary alphabet $\Sigma$%
\begin{vshort}
.
\end{vshort}
\begin{vlong}
:
\begin{fact}
The $F_{\omega^k}$-bounded semi-Thue reachability problem is
$\F_{\omega^k}$-complete, already for $|\Sigma| = 2$.
\end{fact}
\end{vlong}

\paragraph{Reduction.}  Let $\tup{\Upsilon,y,y'}$ be an instance of
size $n>1$ of $F_{\omega^k}$-bounded R[Rewr] over the
two-letters alphabet $\{a_0, a_1\}$.  We build a $(k+2)$-LR[1-bld]
instance in which the rewrite relation $R$ performs the following
sequence:
\begin{enumerate}
\item Weakly compute a budget of size $F_{\omega^k}(n)$, using $\Fw$ described in \autoref{sec-hardy}.\label{c-forward}
\item In this allocated space, simulate the rewrite steps
  of $\Upsilon$ starting from $y$. \label{c-simul}
\item Upon reaching $y'$, perform a reverse Hardy computation using
  $\Bw$ and check that we obtain back the initial Hardy
  configuration.  This check ensures that the lossy rewrites were in
  fact reliable (i.e., no symbols were lost).\label{c-backward}
\end{enumerate}

For Phase~\ref{c-simul}, we define a $\tally$-padded version $\Sim$ of
$\to_\Upsilon$ that works over $\Seqs$:
\begin{align}
  \Sim&\eqdef\{(q_\Sim\sep\sep\sep u\tally^p, q_\Sim\sep\sep\sep
  v\tally^q)\mid u\to_{\Upsilon}v,|u|+p=|v|+q\}\;.
  \shortintertext{This is a length-preserving rational relation. We
    define two more length-preserving rational relations $\Init$ and $\Fin$
    that initialize the simulation with $y$ on the budget space, and
    launch the verification phase if $y'$ appears there, allowing to
    move from Phase~\ref{c-forward} to Phase~\ref{c-simul} and from
    Phase~\ref{c-simul} to Phase~\ref{c-backward}, respectively:}
  \Init &\eqdef \{ (q_\Fw \sep \sep \sep \tally^{\ell+|y|}, q_\Sim \sep
    \sep \sep y\tally^\ell)\mid\ell\geq 0\}\;,\\
  \Fin &\eqdef \{ (q_\Sim \sep \sep \sep y'\tally^\ell, q_\Bw \sep
    \sep \sep \tally^{\ell + |y'|} )\mid\ell\geq 0\}\;.
  \shortintertext{Finally, because
    $F_{\omega^k}(n)=H^{(\Omega_k)_n}(n)$, we define
    our source and target by}
  w&\eqdef q_\Fw \sep \sep a_{k-1}^n\tally \sep \tally^n\:,\qquad\quad
  w'\eqdef q_\Bw \sep \sep a_{k-1}^n\tally \sep \tally^n\;,
\end{align}
and we let $R$ be the 1-bld rational relation $\Fw \cup \Init \cup \Sim
\cup \Fin \cup \Bw$.
\begin{claim}
The given R[Rewr] instance is positive if and only if $\tup{R,w,w'}$
is a positive instance of the $(k+2)$-LR[1-bld] problem.
\end{claim}
\begin{proof}
  Suppose $w\ R_\embedd^\trc\ w'$. It is easy to see that
  the separator symbol ``$\sep$'' and the encodings of states from $Q$
  are reliable. Because of the way the relations treat the states, we
  in fact get 
\begin{equation*}
w\ \Fw_\embedd^\trc\ (q_\Fw\sep\sep\sep\tally^{\ell_1})\ \Init_\embedd\ (q_\Sim
\sep \sep \sep z_1) \ \Sim_\embedd^\trc\ (q_\Sim \sep \sep \sep z_2)
\ \Fin_\embedd\ (q_\Sim \sep \sep \sep
\tally^{\ell_2})\ \Bw_\embedd^\trc\ w'
\end{equation*}
for some strings $z_1, z_2$ and naturals $\ell_1, \ell_2 \in \Nat$. By
Lem.~\ref{lem-corr}, we have $\ell_1\leq F_{\omega^k}(n)$ and $\ell_2\geq
F_{\omega^k}(n)$.  Since $\Init$, $\Sim$, and $\Fin$ are
length-preserving, we get 
\begin{equation}
F_{\omega^k(n)} \geq \ell_1 \geq |z_1| \geq |z_2| \geq \ell_2 \geq F_{\omega^k(n)}
\end{equation} Thus equality holds throughout, and
therefore the lossy steps of $\Sim_\embedd$ in Phase~\ref{c-simul}
were actually reliable, i.e.\ were steps of $\Sim$. This allows us to
conclude that the original R[Rewr] instance was positive.

Suppose conversely that the R[Rewr] instance is positive.  We can
translate this into a witnessing run for $w\ R_\embedd^\trc\ w'$, in
particular, for
$w\ \Fw^\trc\ \comp\ \Init\ \comp\ \Sim^\trc\ \comp\ \Fin\ \comp\ \Bw^\trc\ w'$,
because any successful run from the R[Rewr] instance can be plugged
into the $\Sim^\trc$ phase; Lem.~\ref{lem-corr} and the fact that the
configurations of $\Upsilon$ are bounded by
$F_{\omega^k}(n)$ together ensure that this can be done.
\end{proof}

\subsection{Upper Bound}\label{sub-up}
\paragraph{Well-Structured Transition Systems.}  
As a preliminary, let us show that the lossy rewriting problem is
decidable.  Indeed, the relation $R_\embedd$ can be viewed as the
transition relation of an infinite transition system over the state
space $\Sigma^\ks$.  Furthermore, by Higman's Lemma, the subword
embedding ordering $\embeds$ is a \emph{well quasi ordering} (wqo)
over $\Sigma^\ks$, and the relation $R_\embedd$ is \emph{compatible}
with it: if $u\mathbin{R_\embedd}v$ and $u\embeds u'$ for some
$u,v,u'$ in $\Sigma^\ks$, then there exists $v'$ in $\Sigma^\ks$
s.t.\ $u'\mathbin{R_\embedd}v'$: here it suffices to use $v'=v$ by
transitivity of $\embedd$.

A transition system $\?S=\tup{S,\to,\leq}$ with a wqo $(S,\leq)$ as
state space and a compatible transition relation $\to$ is
called a \emph{well-structured transition system} (WSTS), and several
problems are decidable on such systems under very light effectiveness
assumptions~\citep{abdulla2000c,finkel98b}, among which the
\emph{coverability problem}, which asks given a WSTS $\?S$ and two
states $s$ and $s'$ in $S$ whether there exists $s''\geq s'$
s.t.\ $s\to^\trc s''$.  The lossy rewrite problem when $w\not\embedd
w'$ can be restated as a
coverability problem for the WSTS $\tup{\Sigma^\ks,R_\embedd,\embeds}$
and $w$ and $w'$, since if there exists $w''\embedd w'$ with
$w\mathbin{R_\embedd^\trc} w''$, then $w\mathbin{R_\embedd^\trc} w'$
also holds by transitivity of $\embedd$.

\paragraph{Parameterized Upper Bound.}  In many cases, a \emph{combinatory
  algorithm} can be employed instead of the classical backward
coverability algorithm for WSTS:
we can find a particular coverability witness
$w'=w_0\mathbin{{\embeds}\comp\mathbin{R}^{-1}}w_1\cdots
w_{\ell-1}\mathbin{{\embeds}\comp\mathbin{R}^{-1}}w_\ell\embeds w$ of
length $\ell$ \emph{bounded} by a function akin to $F_{\omega^{k-1}}$ using the
Length Function Theorem of \citep{SS-icalp2011}.  This is a generic
technique for coverability explained in \citep{wqo}, and the reader
will find it instantiated for $(k+2)$-LR[Rat] in App.~\ref{ax-upb}:
\begin{proposition}[Upper Bound]\label{prop-up}
  The problem $(k+2)$-LR[Rat] is in $\F_{\omega^{k+1}}$.
\end{proposition}
The small gap of complexity we witness here with Prop.~\ref{prop-lrr}
stems from the encoding apparatus, which charges us with one extra
symbol.  We have not been able to close this gap; for instance, the
encoding breaks if we try to work without our separator symbol
``$\sep$''.

\section{Applications}\label{sec-rel}
We apply in this section the proof of Prop.~\ref{prop-lrr} to prove
parametric complexity lower bounds for several problems.  In three
cases (propositions \ref{prop-ratep}, \ref{prop-regep}, and \ref{prop-lcs}
below), we proceed by a reduction from the LR problem, but take
advantage of the specifics of the instances constructed in the proof
Prop.~\ref{prop-lrr} to obtain tighter parameterized bounds.  The
hardness proof for the LT problem in Prop.~\ref{prop-term} requires
more machinery, which needs to be incorporated to the construction
of \autoref{sub-low} in order to obtain a reduction.

\paragraph{Rational Embedding.}  We first deal with the classical
embedding problem: We reduce from a $(k+2)$-LR[Rat] instance and use
Prop.~\ref{prop-lrr}.  The issue is to somehow convert an iterated
composition into an iterated concatenation---the idea is similar to
the one typically used for proving the undecidability of PCP.

\begin{proposition}\label{prop-ratep}
  Let $k\geq 2$.  Then $(k+2)$-EP[Rat] is $\F_{\omega^k}$-hard.
\end{proposition}
\begin{proof}
  Assume without loss of generality that $w\neq w'$ in a
  $(k+2)$-LR[Rat] instance $\tup{R,w,w'}$.  We consider sequences of
  consecutive configurations of
  ${\embedd}\comp(R\comp{\embedd})^\tc$ of
  form \begin{equation}\label{eq-seq-ratep} w=v_0\embedd
  u_0\mathbin{R} v_1\embedd
  u_1\mathbin{R}v_2\embedd\cdots\mathbin{R}v_n\embedd
  u_n=w' \end{equation} that prove the LR instance to be positive.
  Let $\$$ be a fresh symbol; we construct a new relation $R'$ that
  attempts to read the $u_i$'s on its first component and the $v_i$'s
  on the second, using the $\$$'s for synchronization:
  \begin{align}\label{eq-rel-ratep}
  R' &\eqdef  \begin{bmatrix}\$w'\$ \\\$ \end{bmatrix}  
  \cdot 
  \left( R \cdot  \begin{bmatrix} \$ \\ \$ \end{bmatrix}  \right)^+ 
  \cdot
   \begin{bmatrix} \varepsilon \\ w\$ \end{bmatrix} 
  \end{align}
  Observe that in any pair of words $(u,v)$ of $R'$, one finds the
  same number of occurrences of the separator $\$$ in $u$ and $v$,
  i.e.\ we can write $u=\$u_n\$\cdots\$u_0\$$ and
  $v=\$v_n\$\cdots\$v_0\$$ with $n>0$, verifying $v_0=w$, $u_n=w'$,
  and $u_i\mathbin{R}v_{i+1}$ for all $i$.  
\begin{vlong}
\par
\end{vlong}
  Assume $u\embeds v$: the embedding ordering is restricted by the
  $\$$ symbols to the factors $u_i\embeds v_i$.  We can
  therefore exhibit a sequence of form \eqref{eq-seq-ratep}.
  Conversely, given a sequence of form \eqref{eq-seq-ratep},
  the corresponding pair $(u,v)$ belongs to $R'\cap{\embeds}$.

  In order to conclude, observe that we can set $\$\eqdef\sep$ in the
  proof of Prop.~\ref{prop-lrr} and adapt the previous arguments
  accordingly, since ``$\sep$'' is preserved by $R$ and appears in both
  $w$ and $w'$ in the particular instances we build.
\end{proof}
\begin{vshort}
\enlargethispage{1.5em}
\end{vshort}
\paragraph{Synchronous Embedding.}
Turning now to the case of synchronous relations, we proceed as in the
previous proof%
\begin{vlong}
 for Prop.~\ref{prop-ratep}
\end{vlong}
, but employ an extra padding
symbol $\pad$ to construct a length-preserving version of the relation
$R$ in an instance of $(k+2)$-LR[Sync], allowing us to apply the
Kleene star operator while remaining regular.
\begin{proposition}\label{prop-regep}
  Let $k\geq 2$.  Then $(k+3)$-EP[Sync] is $\F_{\omega^k}$-hard.
\end{proposition}
\begin{proof}
  Let $\tup{R,w,w'}$ be an instance of $(k+2)$-LR[Sync] with $w\neq w'$ and
  let $\$$ and $\pad$ be two fresh symbols.  Because
  $R\cdot\{(\$,\$)\}$ is a synchronous relation, we can construct a padded
  length-preserving relation
  \begin{align}\label{eq-rel-regep}
    R_\pad&\eqdef\{(u\$\pad^m,v\$\pad^p)\mid m,p\geq 0\wedge(u,v)\in
    R\wedge |u\$\pad^m|=|v\$\pad^p|\}
    \shortintertext{and define a relation similar to \eqref{eq-rel-ratep}:}
    R_\pad'&\eqdef  \begin{bmatrix} \$w'\$ \\ \$ \end{bmatrix}  \cdot R_\pad^\kp \cdot
      \begin{bmatrix} \varepsilon \\ w\$ \end{bmatrix}   \cdot 
       \begin{bmatrix} \varepsilon \\ \pad \end{bmatrix}^{\!\ast}\;.
  \end{align}

  Let us show that $R_\pad'$ is regular: $\{(\$w'\$,\$)\}$ and
  $\{(\varepsilon,w\$)\}$ are relations with bounded length discrepancy
  and $R_\pad^\ast$ is length preserving, thus their concatenation has
  bounded length discrepancy, and can be effectively computed by
  \emph{resynchronization} \citep{saka}.  Suffixing
  $\{(\varepsilon,\pad)\}^\ast$ thus yields a %
  synchronous relation.
  
  As in the proof of Prop.~\ref{prop-ratep}, $R_\pad'$ preserves the
  $\$$ separators, thus if $(u,v)$ belongs to $R'_\pad$, then we can write
  \begin{equation}\label{eq-uv-regep}
    \begin{array}{rlcclcclcccccl}
    u&=\$&u_n&\$&\pad^{m_n}&u_{n-1}&\$&\pad^{m_{n-1}}&\cdots&\$&\pad^{m_1}&u_0&\$&\pad^{m_0}\;,\\
    v&=\$&v_n&\$&\pad^{p_n}&v_{n-1}&\$&\pad^{p_{n-1}}&\cdots&\$&\pad^{p_1}&v_0&\$&\pad^{p_0}\;.
    \end{array}
  \end{equation}
  with $n>0$ and $m_n=0$.  Furthermore, $v_0=w$, $u_n=w'$, and
  $(u_i\$\pad^{m_i},v_{i+1}\$\pad^{p_{i+1}})$ belongs to $R_\pad$,
  thus $u_i\mathbin{R}v_{i+1}$ for all $i$.  If the EP instance is
  positive, i.e.\ if $u\embeds v$, then $u_i\embeds v_i$ and $m_i\leq
  p_i$ for all $i$, and we can build a sequence of form
  \eqref{eq-seq-ratep} proving the LR instance to be positive.
  Conversely, if the LR instance is positive, there exists a sequence
  of form \eqref{eq-seq-ratep}, and we can construct a pair $(u,v)$ as
  in \eqref{eq-uv-regep} above by guessing a sufficient padding amount
  $p_0$ that will allow to carry the entire rewriting.
\begin{vlong}
\par
\end{vlong}
  Finally, as in the proof of Prop.~\ref{prop-ratep}, we can set
  $\$\eqdef\sep$.
\end{proof}

\paragraph{Lossy Termination.}  In contrast with the previous cases,
our hardness proof for the LT problem does not reduce from LR but
directly from a semi-Thue word problem, by adapting the proof of
Prop.~\ref{prop-lrr} in such a way that $R_\embedd^\trc$
is \emph{guaranteed} to terminate.  The main difference is that we
reduce from a semi-Thue system where the length of \emph{derivations}
is bounded, rather than the length of configurations---this is still
$\F_{\omega^k}$-hard since the distinction between time and space
complexities is insignificant at such high complexities.  The
simulation of such a system then builds two copies of the initial
budget in Phase~\ref{c-forward}: a \emph{space} budget, where the
derivation simulation takes place, and a \emph{time} budget, which
gets decremented with each new rewrite of Phase~\ref{c-simul}, and
enforces its termination even in case of losses.  See App.~\ref{ax-term}
for details.
\begin{proposition}\label{prop-term}
  Let $k\geq 2$.  Then $(k+2)$-LT[1-bld] is $\F_{\omega^k}$-hard.
\end{proposition}

\paragraph{Lossy Channel Systems.}
By over-approximating the behaviours of a channel system by allowing
uncontrolled, arbitrary message losses, Abdulla, C\'ec\'e, et
al. \citep{cece95,abdulla96b} obtain decidability results on an
otherwise Turing-complete model.  Many variants of this model have
been studied in the literature
\citep{pepreg,CS-concur08,jks-ifiptcs12}, but our interest here is
that LCSs were originally used as the formal model for the
$\F_{\omega^\omega}$ lower bound proof of \citet{lcs}, rather than a
PEP.

Formally, a \emph{lossy channel system} (LCS) is a finite labeled
transition system $\tup{Q,\Sigma,\delta}$ where transitions in
$\delta\subseteq Q\times
\{?,!\}\times \Sigma\times Q$ read and write on an unbounded channel.
An channel system defines an infinite transition system over its set of
configurations $Q\times\Sigma^\ast$---holding the current state and
channel content---, with transition relation $q,x\to q',x'$ if either
$\delta$ holds a read $(q,{?}m,q')$ and $x=mx'$, or if it holds a
write $(q,{!}m,q')$ and $xm=x'$.  The operational semantics of
an LCS then use the lossy version $\to_\embedd$ of this transition
relation.  In the following, we consider a slightly extended model,
where transitions carry sequences of instructions instead, i.e.\
$\delta$ is a finite set included in
$Q\times(\{{?},{!}\}\times\Sigma)^\ast\times Q$.  The natural decision
problem associated with a LCS is its reachability problem:
\begin{decpb}{Lossy Channel System Reachability}{LCS}
  \item[input] A LCS $\?C$ and two configurations $(q,x)$ and
    $(q',x')$ of $\?C$.
  \item[question] Is $(q',x')$ reachable from $(q,x)$ in $\?C$, i.e.\ does
    $q,x\to^\trc_\embedd q',x'$?
\end{decpb}

The lossy rewriting problem easily reduces to a reachability problem
in a LCS: the LCS \emph{cycles} through the channel contents thanks to
a distinguished symbol, and applies the rational relation at each
cycle; see App.~\ref{ax-lcs} for details.
\begin{proposition}\label{prop-lcs}
  Let $k\geq 2$.  Then $(k+2)$-LCS is $\F_{\omega^k}$-hard.
\end{proposition}

\section{Concluding Remarks}\label{sec-concl}
Post embedding problems provide a high-level packaging of
hyper-Ackermannian decision problems---and thanks to our parametric
bounds, for $k$-Ackermannian problems---, compared to e.g.\
reachability in lossy channel systems (as used in~\citep{lcs}).  The
lossy rewriting problem is a prominent example: because it is stated
in terms of a rational relation instead of a machine definition, it
beneficiates automatically from the theoretical toolkit and multiple
characterizations associated with rational relations.  For a
simple example,
the \emph{increasing} rewriting problem, which employs
$R_\embeds\eqdef{\embeds}\comp R\comp{\embeds}$ instead of
$R_\embedd$, is immediately seen to be equivalent to LR, by
substituting $R^{-1}$ for $R$ and exchanging $w$ and $w'$.

Interestingly, this inversion trick allows to show the equivalence of
the lossy and increasing variants of all our problems, except for
lossy termination:
\begin{decpb}{Increasing Termination}{IT[Rat]}
\item[input] A rational relation $R$ over $\Sigma$ and a word $w$ in
$\Sigma^\ks$.
\item[question] Does $R_\embeds^\trc$ terminate from $w$?
\end{decpb}
A related problem, termination of increasing channel systems with
emptiness tests, is known to be in $\F_3$~\citep{BMOSW} instead of
$\F_{\omega^\omega}$ for LCS termination, but IT[Rat] is more
involved.  Like LR[Rat] or EP[Rat], it provides a high-level
description, this time of \emph{fair termination} problems in
increasing channel systems, which are known to be equivalent to
satisfiability of \emph{safety metric temporal
logic}~\citep{mtl,safeMTL,jenkins}.  The exact complexity of IT[Rat]
is open, with a gigantic gap between the $\F_{\omega^\omega}$ upper
bound provided by WSTS theory, and an $\F_4$ lower bound by
\citet{jenkins}.

\bibliography{journals,conferences,pep}

\begin{thebibliography}{0}
\providecommand{\natexlab}[1]{#1}
\providecommand{\href}[2]{#2}
\providecommand{\nolinkurl}[1]{#1}
\providecommand{\natconfdetails}[2][]{}
\def\UrlBreaks{\do\@\do\\\do\/\do\!\do\|\do\;\do\>\do\]%
 \do\,\do\?\do\'\do+\do\=\do\#}
\providecommand{\urlstyle}[1]{}
\urlstyle{same}

\end{thebibliography}


\begin{thebibliography}{30}
\providecommand{\natexlab}[1]{#1}
\providecommand{\href}[2]{#2}
\providecommand{\nolinkurl}[1]{#1}
\providecommand{\natconfdetails}[2][]{}
\def\UrlBreaks{\do\@\do\\\do\/\do\!\do\|\do\;\do\>\do\]%
 \do\,\do\?\do\'\do+\do\=\do\#}
\providecommand{\urlstyle}[1]{}
\urlstyle{same}

\bibitem[\begingroup Abdulla\endgroup\ and\ \begingroup
  Jonsson\endgroup(1996)Abdulla and Jonsson]{abdulla96b}
Abdulla, P.A. and Jonsson, B., 1996.
\newblock Verifying programs with unreliable channels.
\newblock \emph{Information and Computation}, 127\penalty0 (2):\penalty0
  91--101.
\newblock \href{http://dx.doi.org/10.1006/inco.1996.0053}
  {\nolinkurl{doi:10.1006/inco.1996.0053}}.

\bibitem[\begingroup Abdulla\endgroup\ \textnormal{et~al.}(2000)Abdulla,
  {\v{C}}er{\=a}ns, Jonsson, and Tsay]{abdulla2000c}
Abdulla, P.A., {\v{C}}er{\=a}ns, K., Jonsson, B., and Tsay, Y.K., 2000.
\newblock Algorithmic analysis of programs with well quasi-ordered domains.
\newblock \emph{Information and Computation}, 160:\penalty0 109--127.
\newblock \href{http://dx.doi.org/10.1006/inco.1999.2843}
  {\nolinkurl{doi:10.1006/inco.1999.2843}}.

\bibitem[\begingroup Atig\endgroup\ \textnormal{et~al.}(2010)Atig, Bouajjani,
  Burckhardt, and Musuvathi]{tsoreach}
Atig, M.F., Bouajjani, A., Burckhardt, S., and Musuvathi, M., 2010.
\newblock On the verification problem for weak memory models.
\newblock In \emph{POPL 2010}\natconfdetails[\emph{37th}]{\emph{Annual
  Symposium on Principles of Programming Languages}}, pages 7--18. ACM.
\newblock \href{http://dx.doi.org/10.1145/1706299.1706303}
  {\nolinkurl{doi:10.1145/1706299.1706303}}.

\bibitem[\begingroup Barcel\'o\endgroup\ \textnormal{et~al.}(2012)Barcel\'o,
  Figueira, and Libkin]{barcelo12}
Barcel\'o, P., Figueira, D., and Libkin, L., 2012.
\newblock Graph logics with rational relations and the generalized intersection
  problem.
\newblock In \emph{LICS 2012}\natconfdetails[\emph{27th}]{\emph{Annual IEEE
  Symposium on Logic in Computer Science}}, pages 115--124. IEEE Press.
\newblock \href{http://dx.doi.org/10.1109/LICS.2012.23}
  {\nolinkurl{doi:10.1109/LICS.2012.23}}.

\bibitem[\begingroup Bouyer\endgroup\ \textnormal{et~al.}(2012)Bouyer, Markey,
  Ouaknine, Schnoebelen, and Worrell]{BMOSW}
Bouyer, P., Markey, N., Ouaknine, J., Schnoebelen, {\relax Ph}., and Worrell,
  J., 2012.
\newblock On termination and invariance for faulty channel machines.
\newblock \emph{Formal Aspects of Computing}, 24\penalty0 (4):\penalty0
  595--607.
\newblock \href{http://dx.doi.org/10.1007/s00165-012-0234-7}
  {\nolinkurl{doi:10.1007/s00165-012-0234-7}}.

\bibitem[\begingroup C\'ec\'e\endgroup\ \textnormal{et~al.}(1996)C\'ec\'e,
  Finkel, and {Purushothaman Iyer}]{cece95}
C\'ec\'e, G., Finkel, A., and {Purushothaman Iyer}, S., 1996.
\newblock Unreliable channels are easier to verify than perfect channels.
\newblock \emph{Information and Computation}, 124\penalty0 (1):\penalty0
  20--31.
\newblock \href{http://dx.doi.org/10.1006/inco.1996.0003}
  {\nolinkurl{doi:10.1006/inco.1996.0003}}.

\bibitem[\begingroup Chambart\endgroup\ and\ \begingroup
  Schnoebelen\endgroup(2007)Chambart and Schnoebelen]{pepreg}
Chambart, P. and Schnoebelen, {\relax Ph}., 2007.
\newblock Post embedding problem is not primitive recursive, with applications
  to channel systems.
\newblock In Arvind, V. and Prasad, S., editors, \emph{FSTTCS
  2007}\natconfdetails[\emph{27th}]{\emph{IARCS Annual Conference on
  Foundations of Software Technology and Theoretical Computer Science}}, volume
  4855 of \emph{Lecture Notes in Computer Science}, pages 265--276. Springer.
\newblock \href{http://dx.doi.org/10.1007/978-3-540-77050-3_22}
  {\nolinkurl{doi:10.1007/978-3-540-77050-3_22}}.

\bibitem[\begingroup Chambart\endgroup\ and\ \begingroup
  Schnoebelen\endgroup(2008{\natexlab{a}})Chambart and
  Schnoebelen]{CS-concur08}
Chambart, P. and Schnoebelen, {\relax Ph}., 2008{\natexlab{a}}.
\newblock Mixing lossy and perfect fifo channels.
\newblock In van Breugel, F. and Chechik, M., editors, \emph{CONCUR
  2008}\natconfdetails[\emph{19th}]{\emph{International Conference on
  Concurrency Theory}}, volume 5201 of \emph{Lecture Notes in Computer
  Science}, pages 340--355. Springer.
\newblock \href{http://dx.doi.org/10.1007/978-3-540-85361-9_28}
  {\nolinkurl{doi:10.1007/978-3-540-85361-9_28}}.

\bibitem[\begingroup Chambart\endgroup\ and\ \begingroup
  Schnoebelen\endgroup(2008{\natexlab{b}})Chambart and
  Schnoebelen]{CS-fossacs08}
Chambart, P. and Schnoebelen, {\relax Ph}., 2008{\natexlab{b}}.
\newblock The \(\omega\)-regular {P}ost embedding problem.
\newblock In Amadio, R., editor, \emph{FoSSaCS
  2008}\natconfdetails[\emph{11th}]{\emph{International Conference on
  Foundations of Software Science and Computational Structures}}, volume 4962
  of \emph{Lecture Notes in Computer Science}, pages 97--111. Springer.
\newblock \href{http://dx.doi.org/10.1007/978-3-540-78499-9_8}
  {\nolinkurl{doi:10.1007/978-3-540-78499-9_8}}.

\bibitem[\begingroup Chambart\endgroup\ and\ \begingroup
  Schnoebelen\endgroup(2008{\natexlab{c}})Chambart and Schnoebelen]{lcs}
Chambart, P. and Schnoebelen, {\relax Ph}., 2008{\natexlab{c}}.
\newblock The ordinal recursive complexity of lossy channel systems.
\newblock In \emph{LICS 2008}\natconfdetails[\emph{23rd}]{\emph{Annual IEEE
  Symposium on Logic in Computer Science}}, pages 205--216. IEEE Press.
\newblock \href{http://dx.doi.org/10.1109/LICS.2008.47}
  {\nolinkurl{doi:10.1109/LICS.2008.47}}.

\bibitem[\begingroup Chambart\endgroup\ and\ \begingroup
  Schnoebelen\endgroup(2010{\natexlab{a}})Chambart and
  Schnoebelen]{CS-blockers}
Chambart, P. and Schnoebelen, {\relax Ph}., 2010{\natexlab{a}}.
\newblock Computing blocker sets for the regular {Post} embedding problem.
\newblock In \emph{DLT 2010}\natconfdetails{\emph{International Conference on
  Developments in Language Theory}}, volume 6224 of \emph{Lecture Notes in
  Computer Science}, pages 136--147. Springer.
\newblock \href{http://dx.doi.org/10.1007/978-3-642-14455-4_14}
  {\nolinkurl{doi:10.1007/978-3-642-14455-4_14}}.

\bibitem[\begingroup Chambart\endgroup\ and\ \begingroup
  Schnoebelen\endgroup(2010{\natexlab{b}})Chambart and Schnoebelen]{CS-icalp10}
Chambart, P. and Schnoebelen, {\relax Ph}., 2010{\natexlab{b}}.
\newblock Pumping and counting on the regular {P}ost embedding problem.
\newblock In Abramsky, S., Meyer{ }auf{ }der{ }Heide, F., and Spirakis, P.,
  editors, \emph{ICALP 2010}\natconfdetails[\emph{37th}]{\emph{International
  Colloquium on Automata, Languages and Programming}}, volume 6199 of
  \emph{Lecture Notes in Computer Science}, pages 64--75. Springer.
\newblock \href{http://dx.doi.org/10.1007/978-3-642-14162-1_6}
  {\nolinkurl{doi:10.1007/978-3-642-14162-1_6}}.

\bibitem[\begingroup Elgot\endgroup\ and\ \begingroup Mezei\endgroup(1965)Elgot
  and Mezei]{rtrans}
Elgot, C.C. and Mezei, J.E., 1965.
\newblock On relations defined by generalized finite automata.
\newblock \emph{IBM Journal of Research and Development}, 9\penalty0
  (1):\penalty0 47--68.
\newblock \href{http://dx.doi.org/10.1147/rd.91.0047}
  {\nolinkurl{doi:10.1147/rd.91.0047}}.

\bibitem[\begingroup Fairtlough\endgroup\ and\ \begingroup
  Wainer\endgroup(1998)Fairtlough and Wainer]{fairtlough98}
Fairtlough, M. and Wainer, S.S., 1998.
\newblock Hierarchies of provably recursive functions.
\newblock In Buss, S., editor, \emph{Handbook of Proof Theory}, volume 137 of
  \emph{Studies in Logic and the Foundations of Mathematics}, chapter III,
  pages 149--207. Elsevier.
\newblock \href{http://dx.doi.org/10.1016/S0049-237X(98)80018-9}
  {\nolinkurl{doi:10.1016/S0049-237X(98)80018-9}}.

\bibitem[\begingroup Finkel\endgroup\ and\ \begingroup
  Schnoebelen\endgroup(2001)Finkel and Schnoebelen]{finkel98b}
Finkel, A. and Schnoebelen, {\relax Ph}., 2001.
\newblock Well-structured transition systems everywhere!
\newblock \emph{Theoretical Computer Science}, 256\penalty0 (1--2):\penalty0
  63--92.
\newblock \href{http://dx.doi.org/10.1016/S0304-3975(00)00102-X}
  {\nolinkurl{doi:10.1016/S0304-3975(00)00102-X}}.

\bibitem[\begingroup Friedman\endgroup(1999)Friedman]{friedman99}
Friedman, H.M., 1999.
\newblock Some decision problems of enormous complexity.
\newblock In \emph{LICS 1999}\natconfdetails[\emph{14th}]{\emph{Annual IEEE
  Symposium on Logic in Computer Science}}, pages 2--13. IEEE Press.
\newblock \href{http://dx.doi.org/10.1109/LICS.1999.782577}
  {\nolinkurl{doi:10.1109/LICS.1999.782577}}.

\bibitem[\begingroup Haddad\endgroup\ \textnormal{et~al.}(2012)Haddad, Schmitz,
  and Schnoebelen]{haddad12}
Haddad, S., Schmitz, S., and Schnoebelen, {\relax Ph}., 2012.
\newblock The ordinal-recursive complexity of timed-arc {P}etri nets, data
  nets, and other enriched nets.
\newblock In \emph{LICS 2012}\natconfdetails[\emph{27th}]{\emph{Annual IEEE
  Symposium on Logic in Computer Science}}, pages 355--364. IEEE Press.
\newblock \href{http://dx.doi.org/10.1109/LICS.2012.46}
  {\nolinkurl{doi:10.1109/LICS.2012.46}}.

\bibitem[\begingroup Jan\v{c}ar\endgroup\ \textnormal{et~al.}(2012)Jan\v{c}ar,
  Karandikar, and Schnoebelen]{jks-ifiptcs12}
Jan\v{c}ar, P., Karandikar, P., and Schnoebelen, {\relax Ph}., 2012.
\newblock Unidirectional channel systems can be tested.
\newblock In Baeten, J., Ball, T., and de~Boer, F., editors, \emph{IFIP TCS
  2012}\natconfdetails[\emph{7th}]{\emph{{IFIP} {I}nternational {C}onference on
  {T}heoretical {C}omputer {S}cience}}, volume 7604 of \emph{Lecture Notes in
  Computer Science}, pages 149--163. Springer.
\newblock \href{http://dx.doi.org/10.1007/978-3-642-33475-7_11}
  {\nolinkurl{doi:10.1007/978-3-642-33475-7_11}}.

\bibitem[\begingroup Jenkins\endgroup(2012)Jenkins]{jenkins}
Jenkins, M., 2012.
\newblock PhD thesis, Oxford University.
\newblock In preparation.

\bibitem[\begingroup Karandikar\endgroup\ and\ \begingroup
  Schnoebelen\endgroup(2012)Karandikar and Schnoebelen]{KS-csr12}
Karandikar, P. and Schnoebelen, {\relax Ph}., 2012.
\newblock Cutting through regular {P}ost embedding problems.
\newblock In \emph{CSR 2012}\natconfdetails[\emph{7th}]{\emph{International
  Computer Science Symposium in Russia}}, volume 7353 of \emph{Lecture Notes in
  Computer Science}, pages 229--240. Springer.
\newblock \href{http://dx.doi.org/10.1007/978-3-642-30642-6_22}
  {\nolinkurl{doi:10.1007/978-3-642-30642-6_22}}.

\bibitem[\begingroup Lasota\endgroup\ and\ \begingroup
  Walukiewicz\endgroup(2008)Lasota and Walukiewicz]{ata}
Lasota, S. and Walukiewicz, I., 2008.
\newblock Alternating timed automata.
\newblock \emph{ACM Transactions on Computational Logic}, 9\penalty0
  (2):\penalty0 10.
\newblock \href{http://dx.doi.org/10.1145/1342991.1342994}
  {\nolinkurl{doi:10.1145/1342991.1342994}}.

\bibitem[\begingroup Latteux\endgroup\ and\ \begingroup
  Leguy\endgroup(1983)Latteux and Leguy]{latteux83}
Latteux, M. and Leguy, J., 1983.
\newblock On the composition of morphisms and inverse morphisms.
\newblock In Diaz, J., editor, \emph{ICALP
  1983}\natconfdetails[\emph{10th}]{\emph{International Colloquium on Automata,
  Languages and Programming}}, volume 154 of \emph{Lecture Notes in Computer
  Science}, pages 420--432. Springer.
\newblock \href{http://dx.doi.org/10.1007/BFb0036926}
  {\nolinkurl{doi:10.1007/BFb0036926}}.

\bibitem[\begingroup L\"ob\endgroup\ and\ \begingroup
  Wainer\endgroup(1970)L\"ob and Wainer]{lob70}
L\"ob, M. and Wainer, S., 1970.
\newblock Hierarchies of number theoretic functions, {I}.
\newblock \emph{Archiv {f\"ur} Mathematische Logik und Grundlagenforschung},
  13:\penalty0 39--51.
\newblock \href{http://dx.doi.org/10.1007/BF01967649}
  {\nolinkurl{doi:10.1007/BF01967649}}.

\bibitem[\begingroup Nivat\endgroup(1968)Nivat]{nivat}
Nivat, M., 1968.
\newblock Transduction des langages de {C}homsky.
\newblock \emph{Ann. Inst. Fourier}, 18:\penalty0 339--455.

\bibitem[\begingroup Ouaknine\endgroup\ and\ \begingroup
  Worrell\endgroup(2006)Ouaknine and Worrell]{safeMTL}
Ouaknine, J.O. and Worrell, J.B., 2006.
\newblock Safety {M}etric {T}emporal {L}ogic is fully decidable.
\newblock In Hermanns, H. and Palsberg, J., editors, \emph{TACAS
  2006}\natconfdetails[\emph{12th}]{\emph{International Conference on Tools and
  Algorithms for the Construction and Analysis of Systems}}, volume 3920 of
  \emph{Lecture Notes in Computer Science}, pages 411--425. Springer.
\newblock \href{http://dx.doi.org/10.1007/11691372_27}
  {\nolinkurl{doi:10.1007/11691372_27}}.

\bibitem[\begingroup Ouaknine\endgroup\ and\ \begingroup
  Worrell\endgroup(2007)Ouaknine and Worrell]{mtl}
Ouaknine, J.O. and Worrell, J.B., 2007.
\newblock On the decidability and complexity of {M}etric {T}emporal {L}ogic
  over finite words.
\newblock \emph{Logical Methods in Computer Science}, 3\penalty0 (1):\penalty0
  8.
\newblock \href{http://dx.doi.org/10.2168/LMCS-3(1:8)2007}
  {\nolinkurl{doi:10.2168/LMCS-3(1:8)2007}}.

\bibitem[\begingroup Rose\endgroup(1984)Rose]{rose84}
Rose, H.E., 1984.
\newblock \emph{Subrecursion: Functions and Hierarchies}, volume~9 of
  \emph{Oxford Logic Guides}.
\newblock Clarendon Press.

\bibitem[\begingroup Sakarovitch\endgroup(2009)Sakarovitch]{saka}
Sakarovitch, J., 2009.
\newblock \emph{Elements of Automata Theory}.
\newblock Cambridge University Press.

\bibitem[\begingroup Schmitz\endgroup\ and\ \begingroup
  Schnoebelen\endgroup(2011)Schmitz and Schnoebelen]{SS-icalp2011}
Schmitz, S. and Schnoebelen, {\relax Ph}., 2011.
\newblock Multiply-recursive upper bounds with {Higman}'s {L}emma.
\newblock In \emph{ICALP 2011}\natconfdetails{\emph{International Colloquium on
  Automata, Languages and Programming}}, volume 6756 of \emph{Lecture Notes in
  Computer Science}, pages 441--452. Springer.
\newblock \href{http://dx.doi.org/10.1007/978-3-642-22012-8_35}
  {\nolinkurl{doi:10.1007/978-3-642-22012-8_35}}.

\bibitem[\begingroup Schmitz\endgroup\ and\ \begingroup
  Schnoebelen\endgroup(2012)Schmitz and Schnoebelen]{wqo}
Schmitz, S. and Schnoebelen, {\relax Ph}., 2012.
\newblock \emph{Algorithmic Aspects of WQO Theory}.
\newblock Lecture Notes.
\newblock \href{http://cel.archives-ouvertes.fr/cel-00727025}
  {\nolinkurl{http://cel.archives-ouvertes.fr/cel-00727025}}.

\end{thebibliography}
\clearpage
\appendix
\def\rightmark{Appendices}
\setcounter{page}{1}\renewcommand\thepage{\roman{page}}
\begin{vshort}
\noindent The following appendices contain material that will not be part of the
final paper, if accepted.
\end{vshort}

\section{Unary Alphabet: Prop.~\ref{prop-1ep}}\label{ax-1ep}

We prove here Prop.~\ref{prop-1ep}: $1$-EP[Rat] is in
\textsc{NLogSpace}.

\paragraph{Parikh Images.}  The proof employs the semilinear view of
unary rational relations: a \emph{semilinear set} $S$ is a subset of
$\+Z^k$ described by a finite union of \emph{linear sets}
$(\vec{b},\vec{P})$ defined as $\{\vec{b}+\sum_{i=1}^m
x_i\vec{p}_i\mid x_1,\dots,x_m\in\+N\}$ where $\vec{b}$ is a base in
$\+Z^k$ and $\vec{P}$ is a set of $m$ periods
$\vec{p}_1,\dots,\vec{p}_m$, each in $\+Z^k$.  It is well known that
the \emph{Parikh image} (aka commutative image) $\Psi(L)$ of a regular
language $L$ over an alphabet $\Sigma$ is a semilinear set in
$\+N^{|\Sigma|}$ telling for each symbol of $\Sigma$ how many times it
occurs in a word of $L$.  Formally, let $\Sigma=\{a_1,\dots,a_n\}$,
then a vector $\vec{u}$ is in $\Psi(L)$ iff there exists a word $u$ in
$L$ s.t.\ for all $1\leq i\leq n$, $\vec{u}(i)=|u|_{a_i}$ the number
of occurrences of $a_i$ in $u$.

\begin{proof}[Proof of Prop.~\ref{prop-1ep}]
Given a rational relation $R$ over the unary alphabet $\Sigma=\{a\}$,
we can view its normalized transducer
$\?T=\tup{Q,\Sigma,\Sigma,\delta,I,F}$ as a nondeterministic finite
automaton $\?A=\tup{Q,\Delta,\delta,I,F}$ over the two-letters
alphabet $\Delta=\{(a,\varepsilon),(\varepsilon,a)\}$.  The Parikh
image of $L(\?A)$ is then a semilinear set
$S\subseteq\+N^2$ verifying
\begin{equation}
 S=\{(m,n)\mid (a^m,a^n)\in R\}\;.
\end{equation}

Assume $R\cap{\embeds}\neq\emptyset$, i.e.\
there exists a pair $(m,n)$ in $S$ with $m\leq n$.  Then, there exists
a linear set $(\vec{b},\vec{P})$ in $S$ s.t.\ either
$\vec{b}=(b_1,b_2)$ with $b_1\leq b_2$, or $b_1>b_2$ but there exists
a period $\vec{p}=(p_1,p_2)$ in $\vec{P}$ verifying $p_1<p_2$---and
then there exists $x$ in $\+N$ s.t.\ $b_1+x_1p_1\leq b_2+xp_2$.

It therefore suffices to check in \textsc{NLogSpace} for the existence
of such a non-decreasing basis $\vec{b}$ or such an increasing period
$\vec{p}$ in the normalized transducer $\?T$ for $R$.  This is rather
straightforward:
\begin{itemize}
\item a basis $\vec{b}$ is read along a simple accepting
  run in $\?T$, hence a run of length at most $|Q|$, while
\item a period $\vec{p}$ is read along a simple loop on some state $q$
  of $Q$; we have to check that $q$ is both accessible and
  co-accessible, thus $q$ should lie on an accepting run of length at
  most $2|Q|$ and exhibit a loop of length at most $|Q|$.
\end{itemize}
In both cases it suffices to guess a suitable accepting run to find
such a $\vec{b}$ or $\vec{p}$.
\end{proof}

\section{Codes and Hardy Computations}

\subsection{Pure Codes: Lem.~\ref{lem-pure}}\label{app-pure}
\begin{proof}[Proof of Lem.~\ref{lem-pure}]
  Remember that ordinals are supplied with a \emph{left
  subtraction} operation, because if $\beta\leq \beta'$, then
  $\beta'-\beta$ can be defined as the unique ordinal verifying
  $\beta+(\beta'-\beta)=\beta'$.

  We define an inverse $\pi^{-1}$ by induction on the CNF of ordinals;
  this function yields pure codes exclusively:
  \begin{align}
    \pi^{-1}(0)&\eqdef \varepsilon\;,
  & \pi^{-1}\left(\sum_{i=1}^p\omega^{\beta_i}\right)&\eqdef \beta^{-1}(\beta_p)\tally\pi^{-1}\left(\sum_{i=1}^{p-1}\omega^{\beta_i-\beta_p}\right)\,.
  \qedhere
  \end{align}
\end{proof}

\subsection{Computing with Rational Relations:
  Lem.~\ref{lem-corr}}\label{ax-rules}

\paragraph{Forward and Backward Rules.}  We present here the two
relations $\Fw$ and $\Bw$ under the understanding that they are
suitably restricted to sequences in $\Seqs$.  The relations below are
rational and even 1-bld.  It suffices to give the relations for $\Fw$,
as $\Bw$ just reverses their direction and uses states $q_\Bw$,
$q_{\Bw_1}$, and $q_{\Bw_2}$ instead of $q_\Fw$, $q_{\Fw_1}$, and
$q_{\Fw_2}$.  For $R_0$ given in \eqref{eq-rz0},
\begin{align}
  q_\Fw\sep\sep\tally x\sep\tally^n
  &\;\mathbin{\Fw_0}\;q_\Fw\sep\sep x\sep\tally^{n+1}\label{eq-rFw0}
  \intertext{for all $n$ in $\+N$ and $x$ in $\Sigma_{k\tally}^\ks$.
  Let us repeat the rules for $R_1$ given in
  (\ref{eq-rFw10}--\ref{eq-rFw13}):}
  q_\Fw\sep \sep wa_0\tally x\sep \tally^{n+2}
  &\;\mathbin{\Fw_1}\;q_{\Fw_1}\sep
  w\tally\sep\purify(a_0x)\sep\tally^{n+1}a_0\tag{\ref{eq-rFw10}}\\
  q_{\Fw_1}\sep w\tally^m\sep x\sep\tally^{n+1}a_0^{p+1}
  &\;\mathbin{\Fw_1}\;q_{\Fw_1}\sep w\tally^{m+1}\sep x\sep\tally^{n}a_0^{p+2}
  \tag{\ref{eq-rFw11}}\\
  q_{\Fw_1}\sep w\tally^{m+1}\sep x\sep a_0^{n+2}
  &\;\mathbin{\Fw_1}\;q_{\Fw_1}\sep\sep w\tally^{m+1} x\sep\tally^{n+2}
  \tag{\ref{eq-rFw13}}
  \intertext{where $m,n,p$ range over $\+N$, $w$ over $\Sigma_k^\ks$,
  and $x$ over $\Sigma_{k\tally}^\ks$.  For $R_2$ defined in \eqref{eq-rz2},} 
  q_\Fw\sep \sep wa_i\tally x\sep \tally^{n+2}
  &\;\mathbin{\Fw_2}\;q_{\Fw_2}\sep
  wa_{i-1}\sep\purify(a_ix)\sep\tally^{n+1}a_0\label{eq-rFw20}\\
  q_{\Fw_2}\sep wa_{i-1}^{m}\sep x\sep\tally^{n+1}a_0^{p+1}
  &\;\mathbin{\Fw_2}\;q_{\Fw_2}\sep wa_{i-1}^{m+1}\sep x\sep\tally^{n}a_0^{p+2}
  \label{eq-rFw21}\\
  q_{\Fw_2}\sep wa_{i-1}^{m+1}\sep x\sep a_0^{n+2}
  &\;\mathbin{\Fw_2}\;q_{\Fw_2}\sep\sep wa_{i-1}^{m+1}\tally x\sep\tally^{n+2}
  \label{eq-rFw23}
\end{align}
for $i>0$, $m,n,p$ in $\+N$, $w$ in $\Sigma_k^\ast$, and
$x$ in $\Sigma_{k\tally}^\ks$.%

We define $\Fw = \Fw_0 \cup \Fw_1 \cup \Fw_2$ and analogously for $\Bw$.
The first thing to check is that the reflexive transitive closures of
$\Fw$ and $\Bw$ implement those of $R_H$ and $R_H^{-1}$ as advertised.  A
helpful notion is that of a \emph{phase} of a state $q$, which is a
sequence of rewrites of form
\begin{equation}  
(q_R\sep\sep c_0)\mathbin{R}(q\sep x_1\sep
  c_1)\mathbin{R}\cdots\mathbin{R}(q\sep x_m\sep
  c_{m})\mathbin{R}(q_R\sep\sep c_{m+1})
\end{equation}
for some $c_i$s in $\mathrm{Confs}$ and $x_i$s in
$\Sigma_{k\tally}^\ks$, where $R$ is $\Fw$ or $\Bw$ and thus $q_R$ is
the corresponding state $q_\Fw$ or $q_\Bw$, and
$q$ is an \emph{intermediate} state among
$\{q_{\Fw_1},q_{\Fw_2},q_{\Bw_1},q_{\Bw_2}\}$.  The idea is that a
phase ought to simulate exactly the effect of a single step
$c_0\mathbin{R_H}c_m$ or $c_0\mathbin{R_H^{-1}}c_m$.
\begin{lemma}[Correctness of $\Fw$ and $\Bw$]\label{lem-FwBwcor}
  Let $j$ be in $\{0,1,2\}$ and $c,c'$ be in $\mathrm{Confs}$.  Then
  $(q_\Fw\sep\sep c)\mathbin{\Fw_j^\trc}(q_\Fw\sep\sep c')$ iff $(q_\Bw\sep\sep c')\mathbin{\Bw_j^\trc}(q_\Bw\sep\sep c)$ iff $c \mathbin{R_j^\trc} c'$.
\end{lemma}
\begin{proof}
  The proof is conducted by a case analysis.  Because $\Bw$ is the
  inverse of $\Fw$ with substituted state names, it suffices to show
  the equivalence of $\mbox{$(q_\Fw\sep\sep
    c)$}\mathbin{\Fw_j^\trc}(q_\Fw\sep\sep c')$ with $c \mathbin{R_j^\trc}
  c'$. For $j=0$, the correctness of
  $\Fw_0=\{(q_\Fw{\sep\sep},q_\Fw{\sep\sep})\}\cdot R_0$ is immediate.

  For $j=1$, it suffices to consider a single step of $R_1$, i.e.\ a
  pair of $c=wa_0\tally x\sep\tally^{n+2}$ and
  $c'=w\tally^{n+2}\purify(a_0x)\sep\tally^{n+2}$ with $n$ in $\+N$, $w$ in
  $\Sigma_k^\ks$, and $x$ in $\Sigma_{k\tally}^\ks$.  Then we have a
  rewrite sequence
  \begin{align*}
    (q_\Fw\sep\sep c)
    &\mathbin{\Fw_1}(q_{\Fw_1}\sep w\tally\sep\purify(a_0x)\sep\tally^{n+1}a_0)
    \tag{by \eqref{eq-rFw10}}\\
    &\mathbin{\Fw_1^{n+1}}(q_{\Fw_1}\sep w\tally^{n+2}\sep\purify(a_0x)\sep a_0^{n+2})
    \tag{by \eqref{eq-rFw11}}\\
    &\mathbin{\Fw_1}(q_{\Fw}\sep\sep w\tally^{n+2}\purify(a_0x)\sep\tally^{n+2})
    \tag{by \eqref{eq-rFw13}}\\
    &=(q_{\Fw}\sep\sep c')\;.
  \end{align*}
  Conversely, it suffices to consider a phase of $q_{\Fw_1}$.  It
  is necessarily of the form above, because for \eqref{eq-rFw13} to be
  applicable, the counter segment must be in $\{\tally\}^\ks$, but the
  first step \eqref{eq-rFw10} puts an $a_0$ at the end of the
  segment.  Thus $\Fw_1$ has to go through the appropriate number of
  applications of \eqref{eq-rFw11}.  Therefore, a phase of
  $q_{\Fw_1}$ implies a rewrite of $R_1$.

  We leave the case of $j=2$ as an exercise for the reader, as it is
  very similar to that of $j=1$.
\end{proof}

\subsubsection{Proof of Lem.~\ref{lem-corr}}  The lemma contains
several statements.  The fact that $\Fw$ and $\Bw$ are rational 1-bld
is by definition.  That they are terminating is because they check
that their counters are larger than $1$ in limit steps.
Regarding weak implementation, thanks to
Lem.~\ref{lem-robust} and Lem.~\ref{lem-crct}, we know that
computations using $R_H$ are weak implementations in the sense of
Lem.~\ref{lem-corr}.  Therefore, it remains to prove that the small steps
defined for $\Fw$ and $\Bw$ (i)~correctly implement the rules of
$R_H$ and (ii)~are ``robust'' to losses.

Point~(i) was the topic of Lem.~\ref{lem-FwBwcor}, which in
combination with Lem.~\ref{lem-crct} proves the existence of rewrites 
\begin{align}
  (q_\Fw\sep\sep\pi^{-1}(\alpha)\sep\tally^n)\:&\Fw_\embedd^\trc\:(q_\Fw\sep\sep\sep\tally^m)\label{eq-Fw}\\
  (q_\Bw\sep\sep\sep\tally^m)\:&\Bw_\embedd^\trc\:(q_\Bw\sep\sep\pi^{-1}(\alpha)\sep\tally^n)\label{eq-Bw}
\end{align}
with $m=H^\alpha(n)$.

Turning to point~(ii), in order to prove that a rewrite of
form \eqref{eq-Fw} implies $m\leq H^\alpha(n)$, we want to transform
it into a rewrite according to $(R_H)_\embedd^\trc$, which is known to
imply $m\leq H^\alpha(n)$ thanks to Lem.~\ref{lem-robust} and
Lem.~\ref{lem-crct}.  We conduct a similar proof later
for \eqref{eq-Bw} with $(R_H^{-1})_\embedd^\trc$,
proving \eqref{eq-Bw} to imply $m\geq H^\alpha(n)$.
\begin{claim}
  If $(q_\Fw\sep\sep c)\mathbin{\Fw_\embedd^\trc}(q_\Fw\sep\sep c')$, then
  $c\mathbin{(R_H)_\embedd^\trc}c'$.
\end{claim}
Note that this holds trivially for $\Fw_0$, thus as in the proof of
Lem.~\ref{lem-FwBwcor}, we can focus on \emph{lossy phases} of form
\begin{equation}\label{eq-lphase}
  (q_\Fw\sep\sep c_0)\mathbin{\Fw}(q\sep w_1\sep
  c_1)\mathbin{\Fw_\embedd}\cdots\mathbin{\Fw_\embedd}(q\sep w_m\sep
  c_{m})\mathbin{\Fw}(q_\Fw\sep\sep c_{m+1})
\end{equation}
for some intermediate state $q$.  We will deal with lossy phases of
$q_{\Fw_1}$ here; the proof of the claim for $q_{\Fw_2}$ is similar.
\begin{proof}[Proof of Claim for $\Fw_1$]
  First observe that in a lossy phase like \eqref{eq-lphase} for
  $q_{\Fw_1}$, because \eqref{eq-rFw10} and \eqref{eq-rFw13} are used
  at the beginning and the end of the phase, each intermediate $q\sep
  w_i\sep c_i$ is necessarily in the language
  \begin{equation}\label{eq-def-L1}\begin{array}{rl}
    L_1\eqdef \{q_{\Fw_1}\sep w\tally^{\ell+1}\sep x\sep
    \tally^{n}a_0^{p}\:&\mid w\in\pure(\Sigma_k^\ks),x\in\pure(\Sigma_{k\tally}^\ks),\ell\geq 0,\\&\; p>0,n+p\geq 2\}\;.
  \end{array}\end{equation}

  Define the \emph{atomic embedding} relation over an alphabet $\Gamma$ as
  \begin{align}\label{eq-def-atomic}
    \atomi&\eqdef\{(ww',waw')\mid a\in\Gamma,w,w'\in\Gamma^\ks\}\;.
    \shortintertext{Clearly,}
    \label{eq-atomic}
    {\embeds}&={\atomi^\trc}\;.
  \end{align}
  Moreover, if $y$ and $y'$ are
  two sequences in $L_1$ with $y\embeds y'$, then we can find
  $y_0,y_1,\dots,y_n$ all in $L_1$ s.t.\ $y=y_0\atomi
  y_1\atomi\cdots\atomi y_n$, i.e.\ we can find suitable atomic embeddings
  while remaining in $L_1$.

  Write $\Fw_{\ref{eq-rFw11}}$ for the subrelation of $\Fw$ defined by
  \eqref{eq-rFw11} and $\Fw_{\ref{eq-rFw13}}$ for that of
  \eqref{eq-rFw13}.  Let us show that these subrelations verify
  \begin{align}\label{eq-commutation}
    ({\atoml}\comp\Fw_{\ref{eq-rFw11}})&\subseteq
    (\Fw_{\ref{eq-rFw11}}\comp{\atoml})
    &({\atoml}\comp\Fw_{\ref{eq-rFw13}})&\subseteq
    (\Fw_{\ref{eq-rFw11}}\comp\Fw_{\ref{eq-rFw13}}\comp{\embedd})
  \end{align}
  over $L_1\times L_1$.  This is immediate in most cases, but there
  is a non-trivial case that justifies the use of transitive
  closures in \eqref{eq-commutation} for \eqref{eq-rFw13}:
  \begin{align*}
    q_{\Fw1}\sep w\tally^{m+1}\sep x\sep \tally a_0^{n+2}
  &\atoml q_{\Fw_1}\sep w\tally^{m+1}\sep x\sep a_0^{n+2}\\
  &\mathrm{\Fw_{\ref{eq-rFw13}}}\;q_\Fw\sep\sep w\tally^{m+1} x\sep\tally^{n+2}
  \shortintertext{should be rewritten into}
    q_{\Fw_1}\sep w\tally^{m+1}\sep x\sep \tally a_0^{n+2}
  &\mathrm{\Fw_{\ref{eq-rFw11}}}\;q_{\Fw_1}\sep w\tally^{m+2}\sep x\sep a_0^{n+3}\\
  &\mathrm{\Fw_{\ref{eq-rFw13}}}\;q_{\Fw}\sep\sep w\tally^{m+2} x\sep\tally^{n+3}\\
  &\embedd q_{\Fw}\sep\sep w\tally^{m+1} x\sep\tally^{n+2}\;.
  \end{align*}
  
  To wrap up the proof of the claim, observe that we can apply
  repeatedly \eqref{eq-commutation} to a lossy phase
  like \eqref{eq-lphase} until we have obtained a proper phase of the
  form
  \begin{equation}
   (q_\Fw\sep\sep c_0)\mathbin\Fw_{\ref{eq-rFw10}}(q_{\Fw_1}\sep w_1\sep
   c_1)\mathbin\Fw_{\ref{eq-rFw11}}\cdots\mathbin\Fw_{\ref{eq-rFw11}}(q_{\Fw_1}\sep
  w'_{m'}\sep c'_{m'})\mathbin\Fw_{\ref{eq-rFw13}}\comp{\embedd}\:(q_\Fw\sep\sep
  c_{m+1})\;.
  \end{equation}
  Therefore, by Lem.~\ref{lem-FwBwcor},
  $c_0\mathbin{(R_1)_\embedd}c_m$ as desired.
\end{proof}

Let us turn to the backward version of the claim:
\begin{claim}
  If $(q_\Bw\sep\sep c)\mathbin{\Bw_\embedd^\trc}(q_\Bw\sep\sep c')$, then
  $c\mathbin{(R^{-1}_H)_\embedd^\trc}c'$.
\end{claim}
\begin{proof}
  We proceed as in the proof of the previous claim, by considering
  lossy phases and transforming them into reliable ones.  Focusing on
  $\Bw_1$, the crux of
  the argument mirrors \eqref{eq-commutation} with
  \begin{align}
    (\Bw_{j}\comp{\atoml})
    &\subseteq({\atoml}\comp\Bw_{j})\;,
  \end{align}
  over $L_1\times L_1$ for $j$ in $\{\ref{eq-rFw11},\ref{eq-rFw13}\}$.
  The cases can be solved rather easily thanks to the restriction to
  $L_1$ defined in \eqref{eq-def-L1}.  For instance,
  \begin{align*}
    q_\Bw\sep\sep w\tally^{m+1}x\sep\tally^{n+2}
    &\mathbin{\Bw_{\ref{eq-rFw13}}}\;q_{\Bw_1}\sep w\tally^{m+1}\sep
      x\sep a_0^{n+2}\\
    &\atoml q_{\Bw_1}\sep w\tally^{m}\sep x\sep a_0^{n+2}
    \shortintertext{necessarily has $m>0$ in order to belong to $L_1$,
      thus can be rewritten into}
    q_\Bw\sep\sep w\tally^{m+1}x\sep\tally^{n+2}
    &\atoml q_\Bw\sep\sep w\tally^{m}x\sep\tally^{n+2}\\
    &\mathbin{\Bw_{\ref{eq-rFw13}}}\;q_{\Bw_1}\sep w\tally^{m}\sep x\sep a_0^{n+2}\;.
  \end{align*}
  Similar arguments can be used to complete the proof.  
\end{proof}

\section{Complexity Bounds}

\subsection{Upper Bounds: Prop.~\ref{prop-up}}\label{ax-upb}
\paragraph{Coverability Algorithm.}  
The algorithm for deciding coverability in WSTS is known as
the \emph{backward coverability} algorithm: given an instance
$\tup{R,w,w'}$ with $w\neq w'$, the algorithm starts with the
upward-closure $\embeds(\{w'\})$ of $w'$ as initial set of potential
targets $I_0$.  The algorithm then builds the set of predecessors
$I_1=I_0\cup R_\embedd^{-1}(I_0)=I_0\cup\embeds(R^{-1}(I_0))$: any
sequence that covers $w'$ has to go through $I_1$.  This process is
repeated with $I_{i+1}=I_i\cup\embeds(R^{-1}(I_{i}))$ until
stabilization, which occurs since upward-closed subsets of a wqo
display the
\emph{ascending chain condition}: there exists $n$
s.t.\ $I_{n+1}=I_n$.  As this $I_n$ contains all the words in
$\Sigma^\ks$ that can cover $w'$, it remains to check whether $w$
belongs to the set or not.
This algorithm is effective because, although the sets $I_i$ are
infinite, they can be represented by their $\embeds$-minimal elements,
which are in finite number thanks to the wqo.

\paragraph{Controlled Sequence.}
When moving from decidability issues to complexity ones, we need to
measure the complexity of basic operations in the previous algorithm.
The key computation here is that of a minimal element $u_{i+1}$ in $I_{i+1}$
given a minimal element $u_i$ of $I_i$.  Since $u_{i+1}$ is minimal,
it is produced from some $v_i\embedd u_i$ s.t.\
$u_{i+1}\mathbin{R}v_i$, i.e.\ $u_i=a_1\cdots a_m$ and
$v_i=v'_0a_1v_1\cdots v'_{m-1}a_mv'_m$ for some $a_j$ in $\Sigma$ and
$v'_j$ in $\Sigma^\ks$.

Given $\?T=\tup{Q,\Sigma,\Sigma,\delta,I,F}$ a normalized transducer
for $R$, we know the accepting run with $v_i$ as image is of form
\begin{equation}
  q_0\xrightarrow{(u'_0,v'_0)}q'_0\xrightarrow{(\varepsilon,a_1)}q_1\xrightarrow{(u'_1,v'_1)}q'_1\cdots
  q_{m-1}\xrightarrow{(u'_{m-1},v'_{m-1})}q'_{m-1}\xrightarrow{(\varepsilon,a_m)}q_m\xrightarrow{(u'_m,v'_m)}q'_m
\end{equation}
with $q_0$ in $I$, $q'_m$ in $F$, and $u_{i+1}=u'_0u'_1\cdots
u'_{m-1}u'_m$ as input.  Then none of the segments
$q_j\xrightarrow{(u'_j,v'_j)}q'_j$ can have length greater than $|Q|$,
or $u_{i+1}$ would not be a minimal element of $I_{i+1}$.  Therefore,
$|u_{i+1}|\leq |Q|\cdot(|u_i|+1)$, and any $u_{i+1}$ can be computed
in \textsc{NLogSpace}.

\begin{proof}[Proof of Prop.~\ref{prop-up}] The idea of our
  \emph{combinatory algorithm} is to derive an upper bound on the
  length of a sequence proving reachability.  A nondeterministic
  algorithm can then explore this search space for a witness.

  Assume the $(k+2)$-LR[Rat] instance to be positive.  We consider now
  a sequence of upward-closed sets $\embeds(\{w'\})=I_0\subsetneq
  I_1\subsetneq \cdots\subsetneq I_\ell$ such that $w$ is in $I_\ell$
  but not in $I_i$ for any $i<\ell$, i.e.\ we do not wait for
  saturation of the $I_i$'s but stop as soon as $w$ appears.  We can
  extract a particular minimal element $u_{i+1}$ in each
  $I_{i+1}\setminus I_i$.  Let $g(x)=|Q|\cdot (x+1)$; by the previous
  analysis, $|u_{i+1}|\leq g(|u_i|)$, and of course $|u_0|=|w'|$.  The
  sequence $u_0,u_1,\dots,u_\ell$ is a \emph{bad sequence}: for all
  $i<j$, $u_i\not\embeds u_j$.  By the Length Function
  Theorem~\citep{SS-icalp2011}, the length $\ell$ is bounded by the
  \emph{Cich\'on function} $h_{\omega^{\omega^{k+1}}}((k-1)|w'|)$
  relativized with $h(x)=x\cdot g(x)=|Q|x^2+|Q|x$.

  A nondeterministic algorithm can then set $w_0=w'$ and guess one by
  one a sequence of $\ell$ words $w_i$ in $I_i$ with $w_{i+1}$ in
  $\embeds(R^{-1}(\{w_i\}))$ until $w_\ell\embeds w$.  The space
  required at each step is logarithmic in $|w_i|$, which is bounded
  overall by the Hardy function $h^{\omega^{\omega^{k+1}}}((k-1)|w'|)$
  for the same relativized $h$.

  Finally, given an instance $\tup{R,w,w'}$ of $(k+2)$-LR[Rat] of size
  $n$, as $n\geq|Q|$ we can use $h(x)=x^3+x^2$ instead and bound the
  required space of each step by $h^{\omega^{\omega^{k+1}}}(n)$. The
  space requisites of this algorithm place it in
  $\F_{\omega^{\omega^{k+1}}}$, as the function $h$ is polynomial.
\end{proof}

\section{Applications}

\subsection{Lossy Termination: Prop.~\ref{prop-term}}\label{ax-term}

We prove in this section Prop.~\ref{prop-term}: $(k+2)$-LT[Rat] is 
$\F_{\omega^k}$-hard.

\begin{proof}[Proof Sketch of Prop.~\ref{prop-term}]
  We need for this proof to examine more carefully the construction
  in \autoref{sub-low}.  The following facts are decisive:
  \begin{enumerate}
  \item both $\Fw$ in Phase~\ref{c-forward} and $\Bw$ in
    Phase~\ref{c-backward} are terminating relations,
  \item the simulation of the semi-Thue system $\Upsilon$ in
    Phase~\ref{c-simul} can be
    carried instead with a ``time budget'': it employs sequences of the form
    $\gamma\sep\tally^t$, where $\gamma$ encodes a sequence of $\Seqs$ and
    $t$ tells how many steps are still allowed---initially the same
    budget allocated by Phase~\ref{c-forward}, but
    decremented by $1$ at each rewrite.  This allows to simulate a
    time-bounded semi-Thue system instead of a space-bounded one, but
    they are equivalent as far as $\F_{\omega^k}$ is concerned.
  \end{enumerate}

  Let us detail a bit further the changes we carry.  The new relation
  $R'$ has to be modified to work on words in
  $\Seqs\cdot\{\sep\}\cdot\{\tally\}^\ks$.  The relation $\Fw'$ for
  Phase~\ref{c-forward} needs
  to duplicate its counter increments on both sides of the last separator
  $\sep$ in \eqref{eq-rFw0}, which becomes
  \begin{align}
  q_\Fw\sep\sep\tally x\sep\tally^n\sep\tally^{n}
  &\;\mathbin{\Fw'_0}\;q_\Fw\sep\sep x\sep\tally^{n+1}\sep\tally^{n+1}\;.
  \intertext{The other cases of $\Fw'$ are based on those of
  $\Fw$ and additionally duplicate the contents after
  the last ``$\sep$'': for instance, for \eqref{eq-rFw11}:}
  q_{\Fw_1}\sep w\tally^m\sep x\sep\tally^{n+1}a_0^{p+1}\sep z
  &\;\mathbin{\Fw'_1}\;q_{\Fw_1}\sep w\tally^{m+1}\sep
  x\sep\tally^{n}a_0^{p+2}\sep z\;.
\end{align}
  {The simulation relation $\Sim'$ for Phase~\ref{c-simul}
  now decrements the time budget:}
\begin{align}
  \Sim'&\eqdef \Sim\cdot\begin{bmatrix}\sep\\\sep\end{bmatrix}\cdot\begin{bmatrix}\tally\\\tally\end{bmatrix}^{\!\ks}\cdot\begin{bmatrix}\tally\\\varepsilon\end{bmatrix}\;.
  \intertext{The other relations can be taken to simply preserve the
  time budget:}
  \Bw'&\eqdef\Bw\cdot\begin{bmatrix}\sep\\\sep\end{bmatrix}\cdot\begin{bmatrix}\tally\\\tally\end{bmatrix}^{\!\ks}\;,\\
  \Init'&\eqdef\Init\cdot\begin{bmatrix}\sep\\\sep\end{bmatrix}\cdot\begin{bmatrix}\tally\\\tally\end{bmatrix}^{\!\ks}\;,\\
  \Fin'&\eqdef\Fin\cdot\begin{bmatrix}\sep\\\sep\end{bmatrix}\cdot\begin{bmatrix}\tally\\\tally\end{bmatrix}^{\!\ks}\;.
  \shortintertext{We add a new relation $\mathsf{End}$ that enters an
  infinite loop if the full simulation has been carried:}
  \mathsf{End}&\eqdef\left(\begin{bmatrix}q_\Bw\sep\sep
  a_{k-1}^n\tally\sep\tally^n\sep\\q_{\mathsf{End}}\sep\sep
  a_{k-1}^n\tally\sep\tally^n\sep\end{bmatrix}\cdot\begin{bmatrix}\tally\\\tally\end{bmatrix}^{\!\ks}\right)
  + \left(\begin{bmatrix}q_{\mathsf{End}}\\q_{\mathsf{End}}\end{bmatrix}\cdot\mathrm{Id}_{\Sigma_{k\tally}\uplus\{\sep\}}^\ks\right)\,.
  \intertext{Finally, the source sequence becomes}
  w&\eqdef q_\Fw\sep\sep a_{k-1}^n\tally\sep\tally^n\sep\tally^n\;.
  \end{align}
  The reader can check that the defined relation $R'$ is 1-bld and
  rational, and that the constructed instance $\tup{R',w}$ terminates
  iff the R[Rewr] instance $\tup{\Upsilon,y,y'}$ was positive.
\end{proof}

\subsection{Lossy Channel Systems: Prop.~\ref{prop-lcs}}\label{ax-lcs}

We prove in this section Prop.~\ref{prop-lcs}: $(k+2)$-LCS is
$\F_{\omega^k}$-hard.
\begin{proof}
We reduce from a $(k+2)$-LR[Rat] instance $\tup{R,w,w'}$ and use
Prop.~\ref{prop-lrr}.  Let $\$$ be a fresh symbol and
$\?T=\tup{Q,\Sigma,\Sigma,\delta,I,F}$ the normalized finite
transducer for $R$.

We construct a LCS
$\?C=\tup{Q\uplus\{q_i,q_f\},\Sigma\uplus\{\$\},\delta'}$ that cycles
through its channel content: it starts with $w\$$ as initial channel
contents in some initial state of $\?T$, applies the transitions
$(q,u,v,q')$ of $\?T$ by reading $u$ from the channel and writing $v$
through a transition $(q,?u!v,q')$ of $\delta'$, and cycles back upon
reading $\$$ by transitions $(q,?\$!\$,q')$ in $\delta'$ for all
initial states $q'\in I$ and final states $q\in F$ of $\?T$.  Adding
to $\delta'$ the transitions $(q_i,\varepsilon,q)$ for $q$ in $I$ and
$(q,\varepsilon,q_f)$ for $q$ in $F$, then $(w,w')$ belongs to
$R_{\embedd}^\trc$ iff $q_i,w\$\rightarrow_\embedd^\trc q_f,w'\$$ in
$\?C$.  As in the proof of Prop.~\ref{prop-ratep}, we can tighten this
construction by reusing $\sep$ for $\$$.
\end{proof}

\end{document}